\documentclass[a4paper,onecolumn,11pt,allowtoday,accepted=2025-05-13]{quantumarticle}
\pdfoutput=1
\usepackage[dvipsnames]{xcolor}
\usepackage{amsmath}
\usepackage{amssymb}
\usepackage[normalem]{ulem}
\usepackage{graphicx}
\usepackage{dcolumn}
\usepackage{bm}

\usepackage[left=1in,top=1in,right=1in,bottom=1in]{geometry}
\usepackage[numbers,sort&compress]{natbib}
\usepackage{hyperref}
\hypersetup{
    colorlinks=true,
    linkcolor=blue,
    citecolor=blue,
    }

\usepackage{tikz}
\usepackage{quantikz}
\usepackage{bbold}
\usepackage{relsize}
\usepackage{amsthm}
\usepackage{comment}
\usepackage{svg}
\usepackage{soul}

\newtheorem{thm}{Theorem}
\newtheorem{defn}{Definition}
\newtheorem{cor}{Corollary}
\newtheorem{lem}{Lemma}

\makeatletter
\newcommand*{\rom}[1]{\expandafter\@slowromancap\romannumeral #1@}
\makeatother

\newcommand{\comp}{\mathfrak{comp}}
\newcommand{\dom}{\mathfrak{dom}}

\newcommand{\bs}[1]{\boldsymbol{#1}}

\newcommand{\dg}{\dagger}

\newcommand{\polylog}{\text{polylog}}

\newcommand{\ceil}[1]{\left\lceil #1 \right\rceil}

\newcommand{\norm}[1]{\left \lVert #1 \right \rVert}

\newcommand{\bw}[1]{{\color{brown}\textbf{BW: #1}}}

\begin{document}

\title{Explicit block encodings of boundary value problems for many-body elliptic operators}

\author{Tyler Kharazi}
\email{kharazitd@berkeley.edu}
\affiliation{Department of Chemistry, University of California, Berkeley}%
\affiliation{Berkeley Quantum Information and Computation Center, University of California, Berkeley}

\author{Ahmad M. Alkadri}
\affiliation{Department of Chemical and Biomolecular Engineering, University of California, Berkeley}

\author{Jin-Peng Liu}
\affiliation{Yau Mathematical Sciences Center, Tsinghua University}
\affiliation{Center for Theoretical Physics, Massachusetts Institute of Technology}
\affiliation{Simons Institute and Department of Mathematics, University of California, Berkeley}

\author{Kranthi K. Mandadapu}
\affiliation{Department of Chemical and Biomolecular Engineering, University of California, Berkeley}
\affiliation{Chemical Sciences Division, Lawrence Berkeley National Laboratory, Berkeley}

\author{K. Birgitta Whaley}
\affiliation{Department of Chemistry, University of California, Berkeley}
\affiliation{Berkeley Quantum Information and Computation Center, University of California, Berkeley}

\date{\today}

\begin{abstract}
Simulation of physical systems is one of the most promising use cases of future digital quantum computers. In this work we systematically analyze the quantum circuit complexities of block encoding the discretized elliptic operators {that} arise extensively in numerical simulations for partial differential equations, {including high-dimensional instances for many-body simulations}. When restricted to rectangular domains with separable boundary conditions, we provide explicit circuits to block encode the many-body Laplacian with separable periodic, Dirichlet, Neumann, and Robin boundary conditions, using standard discretization techniques from low-order finite difference methods. To obtain high-precision, we introduce a scheme based on periodic extensions to solve Dirichlet and Neumann boundary value problems using a high-order finite difference method, with only a constant increase in total circuit depth and subnormalization factor. 
{We then present} a scheme to implement block encodings of differential operators 
{acting} on more arbitrary domains, inspired by Cartesian immersed boundary methods. We then block encode the many-body convective operator, which describes interacting particles experiencing a force generated by a pair-wise potential given as an inverse power law of the interparticle distance. This work provides concrete recipes that are readily translated into quantum circuits, with depth logarithmic in the total Hilbert space dimension, that block encode operators arising broadly in applications involving the quantum simulation of quantum and classical many-body mechanics.
\end{abstract}
\maketitle

\newpage
\tableofcontents

\newpage
\section{\label{sec:Intro} Introduction}
Elliptic operators constitute a class of differential operators that generalize the Laplace operator and arise in applications ranging from electrostatics to continuum and quantum mechanics \cite{evans2010partial}. In the context of continuum mechanics,
elliptic operators commonly arise as the spatial operators {within} a time-dependent partial differential equation (PDE), and the related steady-state problem can be solved using a variational principle \cite{teschl2009mathematical}. In quantum mechanics, elliptic operators often arise in quantum chemistry and {in} condensed matter physics as the Hamiltonian operator of the Schr{\"o}dinger equation. In classical settings, the spatial operators in the heat equation and more generally the advection-diffusion equation often correspond to an elliptic operator. Solving these problems often requires a large effort in all but the most 
simple examples, since numerical discretization methods are needed to obtain accurate approximations. Direct discretization methods such as finite differences have a cost that grows exponentially with the spatial dimension and the number of particles. If one discretizes each dimension with $N$ grid points, for an interacting many-body system containing $\eta$ particles in some $d$-dimensional domain $\Omega \subset \mathbb{R}^d$, the dimension of the vector space upon which the discretized operator acts will be of {extremely high} dimension $O(N^{\eta d})$, which becomes intractable for even modest numbers of particles. Quantum algorithms provide a promising approach, as the $N^{\eta d}$ grid points can be stored in the quantum state of $O(\eta d\log(N))$ qubits using the technique of first quantization. We also note that $N$ is not the exponentially growing factor in these simulations, as $N$ is chosen with respect to a fixed accuracy,  however, $N^{\eta d}$ grows exponentially in the number of particles $\eta$ and number of independent spatial dimensions $d$. Therefore, it is expected that quantum advantage will be most pronounced in the regime where $\eta d$ is large, since each additional particle has only a linear increase in memory on the quantum computer, whereas classically the resource requirements grow exponentially with each additional particle. 

The analysis of quantum algorithms for solving linear and nonlinear ordinary \cite{LO08,Ber14,BCOW17,childsQuantumSpectralMethods2020,Kro22,ALWZ22,costaFurtherImprovingQuantum2023} and partial \cite{CPP13,CJS13,MP16,CJO19,LMS20,highPrecision} differential equations is an area of {increasingly} active development. 
A new family {of} approaches for solving linear ordinary differential equations, including those obtained by discretizing a PDE, has been recently developed that uses ideas of ``linear combination of Hamiltonian simulations (LCHS)'' or ``Schr{\"o}dingerization" to approximate the non-unitary time evolution. The Schr{\"o}dingerization technique \cite{jinQuantumSimulationPartial2023} has been used to analyze quantum simulations for solving the Fokker-Planck equation as well as many other kinds of PDEs \cite{jinQuantumSimulationFokkerPlanck2024,jinQuantumSimulationMaxwell23,jinQuantumSimulationBoundary24,jinSchrodingerizationIllPose24}. Proposed at a similar time, \cite{LCHS1} provides an algorithm to implement the LCHS as well as {to} implement fast initial state preparation. This approach was recently improved \cite{LCHS2}, exponentially improving the algorithmic dependence on accuracy over the original LCHS algorithms: {the improvement is also} asymptotically optimal for quantum simulation of non-unitary dynamics.  Despite these advancements, there is {nevertheless} very little evidence for exponential quantum advantage for any end-to-end application in the quantum simulation of PDE's. One reason for this gap is a fundamental lack of resources which describe the basic quantum operations needed to implement the prescribed algorithm. Our goal in this work is to provide an accounting of the computational resources for the block encoding of various differential and multiplication operators. Using low level primitives such as incrementer and reflection circuits, we provide explicit circuit constructions of LCUs for the Laplacian and convective operators under different boundary conditions and estimate the number of Toffoli gates needed to implement the operations. We expect that these subroutines can be {broadly employed in quantum algorithms} for practical simulations of differential equations, quantum and classical.

Quantum algorithms using a block encoding input model will likely be the preferred approach in fault-tolerant quantum computation. {Indeed,} \cite{LCHS2}, and many other near optimal quantum algorithms, use block encoding as the input model. In most cases, the block encoding is {also} applied many times throughout the execution of the algorithm. Therefore, the cost to implement the block encoding is an important factor in determining the overall runtime of the proposed algorithm. A standard approach to implementing block encodings is linear combinations of unitaries (LCU) \cite{childsHamiltonianSimulationUsing2012}. {However,} other approaches based on sparse access models, similar to classical sparse matrix formats, have also been explored \cite{berryEfficientQuantumAlgorithms2007, Lin2022, campsExplicitQuantumCircuits2022}. Combined with quantum signal processing (QSP) \cite{lowOptimalHamiltonianSimulation2017, lowHamiltonianSimulationQubitization2019} or more generally, the quantum singular value transformation (QSVT)  \cite{gilyenQuantumSingularValue2019}, a large class of matrix transformations can be enacted on the block encoded matrix. {These include}  inversion, exponentiation (or time evolution), eigenvalue filtering , and many others\cite{gilyenQuantumSingularValue2019,linOptimalPolynomialBased2020a}. {Such} routines build up a polynomial transformation of the block encoding matrix by performing the block encoding circuit, interleaved with a rotation, a number of times directly related to the degree of the polynomial approximation to the desired function. Therefore, the efficiency of the block encoding is a critical component of the overall efficiency of the proposed quantum algorithm.

The major contribution of this work is the construction and cost analysis of efficient quantum circuits to block encode a broad family of differential operators with nearly arbitrary boundary conditions to high-precision and near-optimal gate complexity and subnormalization.
We focus {here} on the particular problem of block encoding linear operators, 
{however} through the use of the Carleman linearization \cite{liu2021a,An2022,liu2024,foretsExplicitErrorBounds2018,liPotentialQuantumAdvantage2023}, these constructions can be readily extended to represent nonlinear PDEs. In addition to LCHS, these block encodings can be used in conjunction with time-stepping schemes \cite{fangTimemarchingBasedQuantuma} as well as {with} history state schemes \cite{jinTimeComplexityAnalysis2022a} to solve time-dependent PDE's.  {The key results are summarized} in {Table} \ref{tab:circuit complexity}. {We find that} the quantum circuits implementing these block encodings can be done with as few as $O(pd\eta \log(N))$ Toffoli or simpler gates for a separable {differential} operator and {with as few as} $\widetilde{O}((p d \eta)^2 N\log^2(N))$ Toffoli or simpler gates for pairwise interactions with a radially symmetric inverse power law potential $V$, where $N$ is chosen so that for any $\epsilon > 0$, $N^{-p+1}<\epsilon$. These constructions provide an exponential speedup in memory and circuit complexity over classical representations of finite difference discretizations of differential operators. This further implies that quantum algorithms providing exponential speedups for solving a commonly encountered subset of elliptic PDEs using the block encoding framework are retained when including the circuit cost of the block encodings.

\begin{table}
    \centering
    \resizebox{\textwidth}{!}{
    \begin{tabular}{|c|c|c|c|c|c|}
    \hline
    Term & \multicolumn{4}{|c|}{Laplacian (non-interacting)}& Convective Operator\\
    \hline
    Boundary Condition& Periodic      & Dirichlet      & Neumann/Robin           & Periodic Extension   &  Periodic     \\ 
    \hline
    Toffoli complexity &$O(dp\log(N))$  & $O(d\log(N))$     &$O(d\log(N))$   &  $O(dp\log(N))$     & $O((p\eta d)^2 N \log^2(N))$   \\
    \hline
    subnormalization&$O(d)$          & $O(d)$            &$O(d)$          &  $O(d)$             & $O(\eta^2 d ||\nabla V||_{\infty})$ \\
    \hline
    $\#$ancilla     &$O(\log(dp))$   &$O(\log(d))$       &$O(\log(d))$    &  $\widetilde{O}(\log( dp))$     &$O(\log(\log(N)\eta dp))$            \\
    \hline
    accuracy        & $O(N^{-p+1})$  &$O(N^{-2})$        &$O(N^{-1})$     &  $O(N^{-p+1})$      &$O(N^{-p+1})$\\
    \hline
    \end{tabular}
    }
    \caption{
    { Features of the block encodings for the various schemes and boundary conditions employed in this work.} Here $N$ is the number of grid points in $d=1$ spatial dimensions, $p$ is the number of grid points used in the FD stencil, and $\eta$ is the number of particles. Accuracy refers to the rate at which a given discretization converges to the true solution, in terms of $N$.  {The column labels} periodic, Dirichlet, and Neumann refer to the $d-$dimensional Laplacian with the {corresponding} boundary conditions. Periodic extension {refers to} the solution of boundary value problems by periodic extensions of functions {that is introduced} in subsection \ref{subsec:irregular domains}. The {operator}  $\nabla V(x)\cdot \nabla$ is the convective operator given by a scalar pair-wise inverse power law potential $V$ summed over the $\binom{\eta}{2}$ interactions between pairs of particles. The factor of $N$ is related to the polynomial approximation rate to inverse power law potentials. The introduction of a cutoff parameter $\delta$, which we assume scales as $\delta \sim 1/N$, controls the divergence near the singular portion of $V$ and gives the associated scaling (see e.g. Lemma \ref{lem: convRate for invPow}).}
    \label{tab:circuit complexity}
\end{table}

Elliptic operators can be defined in a very general way and encode a great many examples of physical problems. {We use the following mathematical definition to clarify what is meant by an elliptic operator for the sake of this work.}
\begin{defn}[Linear elliptic operator of order $m$]
    Let $L$ be a linear differential operator of order $m$ on a domain $\Omega \subset \mathbb{R}^d$ given by
    \begin{equation*}
        Lu = \sum_{|\alpha|\leq m} a_\alpha(\mathbf{x})\partial^\alpha u \,,
    \end{equation*}
    where $\alpha = (\alpha_1,\ldots,\alpha_d)$ denotes a multi-index, $\partial^\alpha u = \partial_1^{\alpha_1}\cdots \partial_d^{\alpha_d}u$, and here $|\alpha| = \alpha_1 + \dots + \alpha_m$. Then $L$ is called \textit{elliptic} if for every $x\in \Omega$ and every $\bm{v} \neq \mathbf{0} \in \mathbb{R}^d$,
    \begin{equation*}
        \sum_{|\alpha| = m}a_\alpha(x)\bm{v}^\alpha \neq 0 \,,
    \end{equation*}
    where $\bm{v}^\alpha := v_1^{\alpha_1} \dots v_d^{\alpha_d}$.
    \label{def:linear elliptic}
\end{defn}
Of particular interest in the physical sciences are the second-order elliptic operators where $|\alpha| \leq 2$. In the most general second-order case, we may have the rank-2 tensor 
\begin{equation*}
    A(x) = \sum_{i,j=0}^{d-1}a_{ij}(x) \ket{i}\bra{j} \,,
\end{equation*}
with elements $a_{ij}(x) \in \mathcal{C}^2$, we may write the second-order elliptic operator in divergence form 
\begin{equation}
    Lu = - \nabla \cdot \left(A(x)\nabla u\right) 
    \,,
    \label{eq:div form}
\end{equation}
or equivalently in divergence-free form
\begin{equation}
    Lu = \sum_{ij}a_{ij}(x) \frac{\partial^2}{\partial x^i \partial x^j} u + \mathbf{b}(x)\cdot \nabla u 
    \,,
    \label{eq:div free}
\end{equation}
where
\begin{equation*}
    \mathbf{b}_i(x) = \sum_{j=0}^{d-1}\frac{\partial a_{ij}(x)}{\partial x^j}
    \,.
\end{equation*}
In the isotropic case where $A(x) = V(x)$ where $V: \mathbb{R}^d \rightarrow \mathbb{R}$ is a multiplication operator corresponding to scalar multiplication by $V(x)$, this simplifies further to
\begin{equation}
    Lu = V(x)\Delta + \nabla V(x) \cdot \nabla\equiv \sum_{i=0}^{d-1}\left( V(x)\frac{\partial^2}{\partial x^i \partial x^i} + \frac{\partial}{\partial x^i}V(x)\frac{\partial}{\partial x^i} \right)
    \,,
    \label{eq:scalar}
\end{equation}
where we interpret $V(x)$ as a potential function acting on all of the particle position coordinates and $\nabla V(x) \sim \mathbf{b}(x)$ the associated force vector. We focus on the ubiquitous case where $V$ is an scalar function from $\mathbb{R}^d$ to $\mathbb{R}$ and the order of the differential operator is 2, however, generalizations of the schemes we present in this case to vector or tensor functions are straightforward.

\subsection{Prior Work}
\label{subsec:prior work}
Quantum algorithms addressing the complexity of block encoding and simulating systems given by various kinds of elliptic operators {have} been an active area of research. In the work of ~\cite{highPrecision}, explicit block encodings were provided for the $d$ dimensional Laplacian with periodic boundary conditions for arbitrary order finite difference stencils. There has also been some work exploring the application of pseudospectral methods \cite{childsQuantumSpectralMethods2020}, although implementing block encodings of the obtained linear system is likely to be more challenging as the obtained matrix does not enjoy the sparse and regularly repeating structure as the finite difference case. In the work of ~\cite{costaQuantumAlgorithmSimulating2019}, {it was shown} that the Neumann and Dirichlet boundary value problems for the Laplacian on a grid can be implemented by providing an extra ``loop" on boundary nodes of the associated graph Laplacian, however, it was unclear how to implement the corresponding block encoding. In addition, there have been explicit block encodings provided for some sparse matrices such as the graph Laplacian of a d-regular graph~\cite{campsExplicitQuantumCircuits2022}, as well as schemes for the so-called ``hierarchical matrices" considered in~\cite{nguyenBlockencodingDenseFullrank2022}. In the context of real-space simulations,~\cite{kivlichanBoundingCostsQuantum2017} considers the generic costs associated with simulating real-space dynamics resulting from direct domain discretization methods under the influence of a Coulomb potential, where the dependence on the accuracy was later improved in ~\cite{childsQuantumSimulationRealspace2022}. In the context of quantum chemistry {algorithms employing} real (or dual) space grids in first quantization,~\cite{Su2021FirstQuant} provides a very detailed construction of the block encoding of the electronic Hamiltonian using a plane wave discretization. 

{The} work of ~\cite{liEfficientQuantumBlock2023a}  {considers} the task of block encoding the pseudo-differential operator (PDO) representation of an elliptic operator. Similar to pseudo spectral methods, the number of Fourier modes used to approximate the PDO, $N$, can be taken {as} $O(\log(1/\epsilon))$, 
{in contrast to} standard high-order difference methods {that} have only polynomial convergence with the number of stencil points $k$, {since} $N \sim \epsilon^{-1/k}$. 
{For example, section 6.1 of that work shows} how to block encode a variable coefficient elliptic operator of the form $I - \nabla\cdot(V({x})\nabla)$ using $O(dnr(\log(1/\epsilon) +d + n))$ elementary operations, where $r$ is the number of Fourier modes needed to approximate $V$ to the desired accuracy, and $n = \log(N) = \log\log(1/\epsilon)$ is the number of qubits, and $d$ is the spatial dimension. Therefore, when the number of Fourier modes needed to expand {a scalar potential} $V({x})$ is small, this can be very efficient. However, {in} the cases we consider {in this work}, the potential function $V$ is given as an inverse power law of the interparticle distance {and} the number of Fourier modes needed to obtain high precision can be very large, $r \sim O(N)$. Although \cite{liEfficientQuantumBlock2023a}  only considers a periodic boundary condition, replacing the quantum Fourier transform of this algorithm with the discrete cosine and discrete sine transformations should in principle be able to diagonalize the operator with homogeneous Neumann and Dirichlet conditions respectively \cite{ boydChebyshevFourierSpectral2001}. 
However, {the authors of \cite{liEfficientQuantumBlock2023a}} also showed that in the most general cases where the associated PDO symbol is not separable, this approach can have exponentially small success probability. In the {particular} case of a separable symbol, they show, similar to {the findings of the current work,} that the success probability can be $\widetilde{O}(1)$. 
{We have also recently become} aware of \cite{guseynovExplicitGateConstruction2024}, which provides explicit block encodings of Hamiltonian operators given as a linear combination of polynomials in the position and momentum operators in the context of the Schr{\"o}dingerization technique.

\subsection{Our Contributions}
\label{subsec:contrib}
In {the current} work, we provide explicit block encodings of the differential and multiplication operators needed to block encode an elliptic operator on a $d-$dimensional cubic domain with Periodic, Neumann, Dirichlet, and Robin boundary conditions. We provide circuit-level primitives for every part of this work, except for the evaluation of phase factors in block encodings that rely on QSP. We will also use controlled applications of some of the circuit constructions, but do not address their controlled implementation. In Appendix \ref{app:circs}, we provide explicit circuits for the shift and reflection operations used in our LCU, each of which can be implemented using $O(n)$ Toffoli or simpler gates to implement. We show how to efficiently modify the known block encoding of the Laplacian with periodic boundary conditions to enforce Dirichlet, Neumann, and Robin conditions. 
In $d$-dimensions if the boundary is rectangular and the forcing function is separable and satisfies consistency conditions where the boundaries meet, we show how these block encodings can be efficiently combined to block encode the differential equation and boundary value problem in $d$ spatial dimensions. 
We provide explicit gate costs, as well as numerical and asymptotic estimates for the success probability of this procedure for the case of a simple Dirichlet boundary value problem and show that our method of enforcing Dirichlet conditions contributes only a small additive constant to the subnormalization factor over that of the periodic case. 
We then show how a natural extension of the technique we used for regular boundary conditions can be used for boundary conditions imposed on the ``interior" of a computational domain, as might occur when studying flow around an obstacle, or geometries with holes in them (see 
Fig. \ref{fig:non simple domain}). 
We additionally provide a method based on periodic extensions to block encode Dirichlet and Neumann boundary conditions to high precision and furthermore, obtain the same order of accuracy on the boundary as on the interior, at least quadratically improving the global convergence rate over one-sided schemes for approximating the boundary condition. 
As an immediate extension of this approach, we show how to block encode the Laplacian on a more arbitrary domain, assuming access to an oracle that can flag grid points as belonging to the either the interior, exterior, or boundary of the desired domain. We additionally provide a software demonstration of these methods, as well as circuit constructions implemented in Pennylane\cite{bergholm2022pennylaneautomaticdifferentiationhybrid}, which can be found on Github \cite{KharaziCode}.

We then address the block encoding of the so-called convective operator, which takes the form $\nabla V(\mathbf{x}) \cdot \nabla$, for a scalar pair-wise interparticle potential function $V$, evaluated over, $\mathbf{x}$ the multi-particle position coordinates. Because the differential operators involved in the convective term are discretized in a similar manner to the Laplacian terms, our contribution will first focus on constructing the block encoding of a multi-particle scalar potential $V$ and evaluating the multi-particle gradient $\nabla V$. 
We explicitly consider the case when $V(r)$ is a Lennard-Jones potential, a commonly used potential in chemical and molecular dynamics simulations. 
{We note that} this choice of $V$ incorporates many of the salient features of other important scenarios, as it involves the evaluation of an inverse power law function of the inter-particle distance: {such interaction potentials} arise regularly in physical models. We first review the gate cost to block encode a discrete position operator, that we then can use to form a block encoding of a discrete distance function. We discuss the number of queries to this distance function based on the convergence rate of a polynomial approximation to the desired potential function $V$.  We then detail the construction of the block encoding of the gradient operator $\nabla V$, using a finite difference approach. 
Since from Def. \ref{def:linear elliptic}, every elliptic operator for a given scalar potential can always be written as a Laplacian term plus a convective term. Combining the block encoding of the scalar operators $V$ with the block encodings {that} we obtain for the differential operators {is} sufficient to implement block encodings for a very broad class of elliptic PDE's. 

\subsection{Overview}
\label{subsec:overview}

{The remainder of the paper} proceeds as follows. 
Section~\ref{sec:Background} introduces the notation we use throughout the text and briefly reviews the finite difference method and how boundary conditions are enforced. 
We also introduce some of the quantum primitives {that} we use in the text, such as linear combination of unitaries (LCU) \cite{childsHamiltonianSimulationUsing2012} and quantum signal processing (QSP) \cite{lowOptimalHamiltonianSimulation2017}.
 Section~\ref{sec:Laplacian} shows how to block encode the $d-$dimensional Laplacian with various boundary conditions occurring on a separable Cartesian domain $\sim [0,L]^d$. 
We provide examples for block encoding Dirichlet, Neumann, Robin, and mixed boundary conditions and gate complexities to implement the {corresponding} LCUs. 
In addition, we show how to use the method of periodic embedding to obtain high-order {accuracy} schemes for Dirichlet and Neumann boundary conditions in $d$ dimensions. 
We then {present} a method that can be used to block encode the Laplacian on more irregular domains.
Section~\ref{sec:GenElliptic} introduces the primitives needed to block encode the convective term $\nabla V(x) \cdot \nabla$, including the evaluation of multi-particle potential functions given by an inverse power law. 
We provide circuits to compute the squared distance between any two grid points in $d$ dimensions and then bound the polynomial degree needed to obtain an $\epsilon-$accurate approximation of $V$ in terms of a cutoff parameter.
Finally, we provide an alternative approach to uniform polynomial approximation by approximating the function with different polynomials close to and far away from the singular point in the potential that can be used to practically reduce the dependence on the cutoff in the numerical approximation of the desired function.
{Section \ref{sec:Discussion} concludes with a discussion and outlook.}

\section{\label{sec:Background} Background}
\subsection{Notation}
We will use $d$ to refer to the underlying spatial dimension of the differential operator, $\eta$ to refer to the number of particles, and $N$ the number of grid points used in each dimension. We will commonly use $p = 2a+1$ to be the number of grid points used in the stencil for the Laplacian and $2a$ the number of grid points for the first derivative. We use $r_j, \; 0\leq |j| \leq a$ to refer to the central difference coefficients for the $p$th order finite difference stencil of the second derivative, and $c_j, \;1\leq |j| \leq a$ {to refer to} the central difference coefficients for the $p$th order finite difference stencil of the first derivative. We will also generally use the convention that upper indices are used to index particles and lower indices are used to index spatial dimensions. For example, if $\mathcal{L}$ is some linear differential operator for $d=1$, then $\mathcal{L}^i_j$ refers to the action of $\mathcal{L}$ on particle $i\in[\eta]$ in dimension $j\in[d]$ where $[N] := \{0,\ldots,N-1\}$.

We will use the notation $\norm{\cdot}$ to refer to {the} standard $2$-norm when the argument is a vector and {to} the induced $2$-norm if the argument is a matrix.   We use the notation $\bs{\alpha}$ to refer to a multi-index. We will write $\alpha_i^j \in [N]$ to refer to an arbitrary grid point for particle $j$ along direction $i$, so $\bs{\alpha}\in[N^{\eta d}]$. When the meaning is obvious, we may also use the simplified indicial notation $i \in [N^{\eta d}]$ to refer to the index given by the map $\alpha_j^i \in [N]  \rightarrow  \alpha_j^iN^{d i + j} \in [N^{\eta d}]$. We use the notation $\ket{x}\ket{y}$ to refer to the tensor product of two $n$-qubit states. The notation $\ket{\bs{\alpha}}, \bs{\alpha} \in [N]^d$ should be interpreted to mean $\ket{\alpha_0}\ket{\alpha_1}\cdots\ket{\alpha_{d-1}}$, where each $\alpha_i \in [N]$. Therefore, the $\eta$ particle $d$-dimensional Hilbert space is given by the span over states of the form $\ket{\bs{\alpha}^0}\ket{\bs{\alpha}^1}\cdots\ket{\bs{\alpha}^{\eta -1}}$, where each $\ket{\bs{\alpha}^i}$ is given as a $N^d$-dimensional state vector so that the total dimension of the underlying Hilbert space is $N^{\eta d}$.

We also use the following asymptotic notation 
\begin{equation}
    \begin{aligned}
        f \in O(g) &\iff \frac{f(n)}{g(n)}\rightarrow 0,\\
        f \in \widetilde{O}(g) &\iff f \in O(g\polylog(g)),\\
    \end{aligned}
    \label{def:big O}
\end{equation}
and write $x(n) \rightarrow y$ to mean that $x(n)$ tends to $y$ as n grows large.
\subsection{Discretization Scheme}
\label{subsec:dom disc}
Finite difference methods (FDM)s are a robust technique for discretizing spatial derivative operators on Cartesian domains. We use the FDM by first discretizing the problem domain $\Omega \subset \mathbb{R}^d$ into regular grid points. If each of the $d$ spatial dimensions are discretized with $N$ grid points, then the domain $\Omega$ becomes associated with the set of grid points $x_{\bs{\alpha}} \in \mathbb{Z}_N^d$, where $\bs{\alpha}$ is a multi-index indexing all the grid points in the $N^{d}$ dimensional discretized position space and $x_{\bs{\alpha}}$ is the map from the grid index onto the physical domain.

For example, in one spatial dimension, the first derivative operator can be approximated to second order using a central-difference discretization scheme as
\begin{equation}
    \frac{d}{dx}f(x) \rightarrow \frac{f(x+h/2) - f(x-h/2)}{h/2} + O(h^2),
    \label{eq:FD 1d}
\end{equation}
where $h = 1/N$.
In addition, an arbitrary function $f$ evaluated on the induced grid results in a $N$-dimensional vector
\begin{equation}
    f(x) \rightarrow \sum_{i\in[N]} f(x_i)\ket{i}.
    \label{eq:vector fn}
\end{equation}
If we consider the second-order PDE with periodic boundary conditions
\begin{align*}
    \frac{d^2}{dx^2}u(x) &= f(x) \hspace{.2cm} 0<x<1\\
    u(0) &= u(1),
\end{align*}
discretization of the interval $[0,1]$ into $N$ grid points and using a $3$ point central-difference scheme to approximate the second derivative results in the linear system
\begin{equation}
\begin{aligned}
      \frac{1}{h^2}\begin{pmatrix}
        -2 & 1 &0& \cdots & 0 & 1\\
        1 & -2 &1& \cdots & 0 & 0\\
        \vdots &\ddots & \ddots & \ddots& & \vdots\\
        0 &\cdots & 1 &-2 &1 & 0\\
        0 &\cdots &0 & 1 & -2& 1\\
        1 &\cdots& 0& 0 & 1 & -2
    \end{pmatrix}
    \begin{pmatrix}
        u_0\\
        u_1\\
        \vdots\\
        u_{N-3}\\
        u_{N-2}\\
        u_{N-1}
    \end{pmatrix}
    &=
    \begin{pmatrix}
        f_0\\
        f_1\\
        \vdots\\
        f_{N-3}\\
        f_{N-2}\\
        f_{0}
    \end{pmatrix}\\
    L\hspace{2.5cm} \mathbf{u}\hspace{.6cm} &= \hspace{.4cm}\mathbf{f}
\end{aligned}
    \label{eq: 1d d/dx}
\end{equation}
We define $L$ to be the matrix on the left hand side of \eqref{eq: 1d d/dx} and $\mathbf{u}$, $\mathbf{f}$ are vectors corresponding to $u(x)$ and $f(x)$ respectively in the sense of \eqref{eq:vector fn}.

We can generalize {the above definitions of differential operators on lattices} to higher-order derivatives, by using additional grid points in the evaluation of the differentials. The number of stencil points is closely connected to the convergence rate of the discretization scheme, and thereby the number of grid points needed to solve a particular problem to a fixed error tolerance. The following theorem characterizes the error in implementing a $p=2a+1$ point finite difference stencil on $N$ grid points in terms of the smoothness of the solution.
\begin{lem}{[Error in 2a+1 point stencil on $N$ grid point Laplacian (\cite{kivlichanBoundingCostsQuantum2017} Theorem 7)]}
    Let $\psi(x) \in \mathcal{C}^{2a+1}$, the space of functions with continuous derivatives up to order $2a+1$, be a solution to the PDE, then the error in the (2a+1)-point centered difference formula for the second derivative of $\psi(x)$ evaluated on a uniform mesh with spacing $h = 1/N$ is at most
    \begin{equation}
        \frac{\pi^{3/2}}{9}e^{2a(1-\ln(2))}h^{2a-1}\max_x\left|\frac{\partial\psi(x)}{\partial x^{2a+1}}\right|.
    \end{equation}
    \label{thm:fd conv stencil}
\end{lem}

Computing finite difference stencil coefficients for regularly spaced grid points is fairly routine and an analytic form of the coefficients are known to arbitrary order; a tabulated resource can even be found on Wikipedia \cite{FiniteDifferenceCoefficient2023}. Research on the computation of finite difference coefficients in more arbitrary settings can be found, e.g., in Ref. \cite{liGeneralExplicitDifference2005}, and in general can be obtained via the theory of interpolating polynomials. For the second derivative with a $p$-point central difference scheme, the coefficients are given by
\begin{equation}
    r_j \equiv \begin{cases}
    \frac{2(-1)^{j+1}(k!)^2}{j^2(a-j)!(a+j)!}& j \in \{1,\ldots,a\}\\
    -2\sum_{j=1}^a r_j & j = 0\\
    r_{-j} & j \in \{-a, \ldots, -1\}.
    \end{cases}
    \label{eq:FD coeffs lap}
\end{equation}
Similarly, for the first derivative we have
\begin{equation}
    c_j \equiv \begin{cases}
    \frac{(-1)^{j+1}((2a)!)^2}{j(2a-j)!(2a+j)!} & 1\leq |j| \leq a\\
    0 & j =0.
    \end{cases}
    \label{eq:FD coeffs d/dx}
\end{equation}

In addition to the promise on the accuracy of high-order finite difference stencils given the smoothness of the solution, we are also given a promise on the smoothness of the solution provided a smooth forcing function $f$ and potential function $V$. This is characterized by the \textit{elliptic regularity theorem}\cite{evans2010partial} for linear, second-order elliptic PDEs.
\begin{lem}{[Elliptic Regularity Theorem, Ref. \cite{evans2010partial} Section 6.3 Theorem 5]}
    Let $m$ be a nonnegative integer and assume $\mathbf{b}(x) \in C^{m+1}(\Omega)$ and $f \in H^m(\Omega)$, the Sobolev space of function satisfying $\norm{f}_{H^m(\Omega)} < \infty$ where $\norm{\cdot}_{H^m(\Omega)}$,
    \begin{align}
        \|u\|_{H^m(\Omega)} := \left(\sum_{|\alpha| \le m} \|D^\alpha u \|_{L^2(\Omega)}^2\right)^{1/2} 
        \label{eq:Sobolev-norm}
    \end{align}
    and
    \begin{align}
        D^\alpha &:= \frac{\partial^{|\alpha|}}{\partial x_1^{\alpha_1} \dots \partial x_1^{\alpha_d}}  
        \quad , \quad 
        |\alpha| = \alpha_1 + \dots + \alpha_d \,.
        \label{eq:D-alpha}
    \end{align}
     Now, suppose $u$ solves the second order linear boundary value problem
    \begin{equation}
    \begin{aligned}
        (\Delta + \mathbf{b}(x)\cdot \nabla)u &= f, \hspace{.2cm} \text{in } \Omega\\
        u &= 0 \hspace{.2cm} \text{in } \partial\Omega,
    \end{aligned}
    \end{equation}
    where the boundary $\partial \Omega \in C^{m+2}$. Then, the solution $u \in H^{m+2}(\Omega)$, and we have the estimate that 
    \begin{equation}
        \norm{u}_{H^{m+2}(\Omega)} \leq C \left(\norm{f}_{H^{m}(\Omega)} + \norm{u}_{L^2(\Omega)}\right)
    \end{equation}
    for some constant $C$ depending only on $m,$ $\Omega$, and $\mathbf{b}$.
    \label{lem:ellip regularity}
\end{lem}

Given a smooth $\mathbf{b} = \nabla V$ and well-behaved right hand side forcing function $f$ over a domain with smooth boundary, we are guaranteed a solution to the PDE that is also smooth, so that the estimate given in Lemma \ref{thm:fd conv stencil} provides a useful bound on the norm of the largest derivative of the solution. Using high-precision schemes on a quantum computer may seem counter intuitive at first, as we can represent exponentially many grid points with a linear number of qubits. However, in many cases the underlying differential operator has a norm or condition number dependence that is polynomial in the number of grid points $N$. Therefore, simply using exponentially many grid points is not an effective strategy as the cost to perform a linear solve or time-stepping scheme will go as a polynomial of $N$, therefore it is unlikely for quantum computers to achieve an exponential speedup with respect to the number of grid points in realistic applications. By contrast, a high-precision scheme only increases the total norm and conditioning by a constant, while exponentially decreasing the total number of grid points needed to approximate the operator to a given accuracy. Therefore, it is useful to consider the higher-order schemes for approximating the differential operator with finite differences on a quantum computer. When considering high-dimensional applications, the advantages of using a high-precision scheme on a quantum computer are more pronounced.

In many spatial dimensions and for multiple particles, the operators can be constructed using a one-local tensor product structure. That is, the differential operator in $d$ spatial dimensions (with each dimension satisfying the same boundary conditions and number of grid points) can be represented as a direct sum over $d$ copies of the $d=1$ case
\begin{equation}
    \mathcal{D} = \sum_{i\in[d]} \mathcal{D}_i,
\end{equation}
where
\begin{equation}
    \mathcal{D}_i =  I_N^{\otimes i}\otimes D \otimes I_{N}^{\otimes (d - i - 1)}.
\end{equation}
Although the dimension of each matrix is $O(N^d)$, the total matrix norm scales linearly with $d$ through a straightforward application of the triangle inequality applied to the above relations.
This construction is readily generalized to the case where the dimensions have different boundary conditions and {different} numbers of nodes, {as} long as the boundary conditions are consistent and can be expressed in terms of each coordinate {in a separable manner.} 
We similarly define the second derivative scalar operator in $d$ dimensions by \begin{equation}
    \mathcal{L} = \sum_{i\in[d]} \mathcal{L}_i,
\end{equation}
where
\begin{equation}
    \mathcal{L}_i = I_N^{\otimes i}\otimes L \otimes I_{N}^{\otimes (d - i - 1)}.
\end{equation}

In the {situation} of many particles, $\eta$, occupying positions in $d$-dimensions, the finite difference operators satisfy a similar generalization as the continuum {many-particle} differential operators. In this case, each particle is associated with an $N^{d}$-dimensional subspace of the $N^{d\eta}$-dimensional Hilbert space, where the particle-localized differential operators for particle $j$ and direction $i$ satisfy
\begin{equation}
    \begin{aligned}
    \mathcal{D}_i^j &= I_{N^d}^{\otimes j}\otimes \mathcal{D}_i \otimes I_{N^d}^{\otimes (\eta - j-1)}\\
    \mathcal{L}_i^j &= I_{N^d}^{\otimes j}\otimes \mathcal{L}_i \otimes I_{N^d}^{\otimes (\eta - j-1)}.
    \end{aligned}
    \label{eq:eta-d-dim mats}
\end{equation}
Therefore, the block diagonal matrix encoding all of the first and second derivatives can be expressed {as}
\begin{equation}
    \begin{aligned}
    \bs{D} &= \sum_{j\in[\eta]}\sum_{i\in [d]}\mathcal{D}_i^j \otimes \ket{ij}\bra{ij}\\
    \bs{L} &= \sum_{j\in[\eta]}\sum_{i\in [d]}\mathcal{L}_i^j \otimes \ket{ij}\bra{ij},
    \end{aligned}
    \label{eq:total deriv mats}
\end{equation}
and the associated scalar operators can be obtained by computing the partial trace over the indices $i$ and $j$.

\subsection{Encoding into quantum states}
\label{subsub:quantum encoding}
We consider an encoding where each basis state is associated with a particular location on the lattice. In each dimension we have that the $N$ grid points are associated with $n \equiv \ceil{\log(N)}$ qubits. Therefore in $d$ dimensions, we need $d n$ qubits to represent the state space of a single particle, and $\eta d n$ qubits to represent $\eta$ particles in $d$ dimensions. As an example, we can define a single particle position basis state in $d$-dimensions as the tensor product
\begin{equation}
    \ket{\bs{\alpha}^0} = \ket{\alpha^0_0}\ket{\alpha^0_1}\cdots\ket{\alpha^0_{d-1}},
\end{equation}
where $\alpha^0_i \in [N]$ and $i \in [d]$.
The elements in this tensor product space are evaluated corresponding to the indexing relations
\begin{equation}
    \bs{\alpha}^0 \in [N]^d \rightarrow  \sum_{i=0}^{d-1}\alpha^0_i N^i \in [N^d],
\end{equation}
where $[N]^d$ is the set of $d$-tuples with entries in $[N]$.

Similarly, for $\eta$-particles in $d$-dimensions, we represent the {state} space as
\begin{equation}
    \ket{\bs{\alpha}} = \ket{\bs{\alpha}^0}\ket{\bs{\alpha}^1}\cdots\ket{\bs{\alpha}^{\eta-1}},
\end{equation}
where $\bs{\alpha}^i \in [N]^{d}$ and $i \in [\eta]$. We treat $\ket{\bs{\alpha}}$ as an arbitrary basis state in our many-body Hilbert space.

Furthermore, we can map any function $f:\mathbb{R}^{\eta d}\rightarrow \mathbb{R}$ to a corresponding state on the lattice. Therefore upon discretization, the function $f$ obtains the {following} representation as a vector
\begin{equation}
    \ket{f} = \sum_{\bs{\alpha}} f({\bs{\alpha}})\ket{\bs{\alpha}},
\end{equation}
which would need to be suitably $l^2$-normalized in order to correspond to a valid quantum state.
Similarly, a function $f:\mathbb{R}^{\eta d}\rightarrow \mathbb{R}$ corresponding to the operation of scalar multiplication is implemented as the diagonal matrix
\begin{equation}
    \mathbf{F} = \sum_{\bs{\alpha}}f(\bs{\alpha})\ket{\bs{\alpha}}\bra{\bs{\alpha}},
\end{equation}
which to be implemented as a block encoding will need to be normalized by the factor $\max_{\bs{\alpha}} |f(\bs{\alpha})|$, to ensure $\norm{\mathbf{F}}_1 \leq 1$.

\subsection{Discrete boundary value problems}
\label{subsec:bkgrnd bvp}
We focus on the three major kinds of boundary conditions encountered {in physical problems, namely} Dirichlet, Neumann, and Robin conditions. The Dirichlet condition specifies the value of the solution at particular subsets of the domain, the Neumann condition specifies the value of the derivative of the solution at particular subsets of the domain, and the Robin condition specifies a weighted linear combination of the value and the derivative of the function. Using a  discretization of $[0,1]$ into $N=1/h$ grid points we can express the linear system corresponding to the Dirichlet boundary value problem as {follows}. For simply connected domains where the boundary conditions are specified on the end-points, we can ``fold in" the boundary value problem into the interior to obtain a higher-order accuracy scheme. Letting $a$ and $b$ be arbitrary real-valued constants describing the boundary values, and applying this approach to the Dirichlet boundary value problem can be expressed as the symmetric linear system
\begin{equation}
\frac{1}{h^2}
   \begin{pmatrix}
        -2 & 1 & 0 &\cdots & 0\\
        1 & -2 & 1&\cdots & 0\\
        \vdots& \ddots &\ddots &\ddots &\vdots\\
        0&\cdots & 1 &-2 &1\\
        0&0 &\cdots & 1 & -2
    \end{pmatrix}
    \begin{pmatrix}
        u_0\\
        u_1\\
        \vdots\\
        u_{N-2}\\
        u_{N-1}
    \end{pmatrix}
    =
        \begin{pmatrix}
        f_0 - \frac{a}{h^2}\\
        f_1\\
        \vdots\\
        f_{N-2}\\
        f_{N-1} - \frac{b}{h^2}
    \end{pmatrix}.
\label{eq:dir mat}
\end{equation}
Using $a$ and $b$ to also refer to the boundary values for the Neumann condition, i.e. $u'(0) = -a$ (where the minus sign reflects the convention of outward facing normal vectors) and $u'(1) = b$, we extend the system by two grid points on either end and we obtain the system
\begin{equation}
    \frac{1}{h^2}
    \begin{pmatrix}
        1 & -1 & 0 &\cdots & 0\\
        -1 & 2 & -1&\cdots & 0\\
        \vdots& \ddots &\ddots &\ddots &\vdots\\
        0&\cdots & -1 &2 &-1\\
        0&0 &\cdots & -1 & 1
    \end{pmatrix}
        \begin{pmatrix}
        u_{-1}\\
        u_0\\
        \vdots\\
        u_{N-1}\\
        u_{N}
    \end{pmatrix}
=
\begin{pmatrix}
    \frac{a}{h}\\
    -f_0\\
    \vdots\\
    -f_{N-1}\\
    -\frac{b}{h}
\end{pmatrix},
\label{eq:neum mat}
\end{equation}
and, more generally, a Robin condition of the form $au(0) + bu'(0) = -\alpha$ and $cu(L) + d u'(L) = \beta$ with $\alpha, \beta \in \mathbb{R}$ can be expressed {by}
\begin{equation}
\frac{1}{h^2}
    \begin{pmatrix}
        1 - ah/b & -1 & 0 &\cdots & 0\\
        -1 & 2 & -1&\cdots & 0\\
        \vdots& \ddots &\ddots &\ddots &\vdots\\
        0&\cdots & -1 &2 &-1\\
        0&0 &\cdots & -1 & 1 + ch/d
    \end{pmatrix}
    \begin{pmatrix}
        u_{-1}\\
        u_0\\
        \vdots\\
        u_{N-1}\\
        u_{N}
    \end{pmatrix}
=
\begin{pmatrix}
    \frac{\alpha}{b}\\
    -f_0\\
    \vdots\\
    -f_{N-1}\\
    -\frac{\alpha}{d}
\end{pmatrix}.
\label{eq:rob mat}
\end{equation}

The extension to so-called ``mixed" boundary conditions, where different types of boundary conditions are specified on independent subsets of the domain is straightforward. In $d=1$ for example, we could enforce a mixed boundary value problem with Dirichlet condition on the left-most endpoint and Neumann condition on the right-most endpoint by adjusting the periodic matrix to match the Dirichlet matrix in the first row and the Neumann matrix in the last row. However, one difficulty of applying boundary conditions using the above approach is that it is unclear how to approximate them to high-precision, and so the above schemes do not readily generalize to high-precision case as easily as the periodic case. We address this concern in more detail in sec \ref{subsec:periodic extension} when we consider high-order schemes for Dirichlet and Neumann boundary value problems.

Since we assume that the boundary conditions can be described in each dimension independently, the corresponding discretization can be constructed through the local tensor product structure we discuss in \ref{subsec:dom disc}, where each $1d$ matrix encodes a boundary value problem similar to those given above. That is, for each dimension we can construct {an operator taking the form of one of the matrices in ~\eqref{eq:dir mat}-\eqref{eq:rob mat} including} the specified boundary conditions,
{then} we can combine {the operators from each dimension} with a direct sum to represent the entire system matrix with the desired boundary conditions {overall}.  We focus on the particular case of block encoding the matrix operator in $d$ dimensions subjected to various boundary conditions and assume access to an efficient quantum circuit {that} prepares normalized versions of the right-hand-side vectors {in~\eqref{eq:dir mat} -\eqref{eq:rob mat}.}

\subsection{Block encoding and LCU implementation}
Block encoding is an access model for matrix operations on a quantum computer. A block encoding is characterized by 3 parameters, {namely,} the subnormalization factor $\lambda$, the number of ancilla qubits $q$, and the accuracy of the block encoding $\epsilon$.
\begin{defn}[$(\lambda, q, \epsilon)$ Block Encoding]
    For an $n$ qubit matrix $A$, we say that the $n+q$ qubit unitary matrix $U_A$ is {a} $(\lambda, q, \epsilon)$ block encoding of $A$ if
    \begin{equation}
        ||A - \lambda (\bra{0}^{\otimes q}\otimes I_n) U_A (\ket{0}^{\otimes q}\otimes I_n)|| \leq \epsilon \,.
    \end{equation}
    If we can implement the above unitary $U_A$ exactly, then we call it {a} $(\lambda, q)-BE(A)$ block encoding of $A$.
    \label{def: Block Encoding}
\end{defn}
In this work, we will assume that the block encoding of the matrix is obtained exactly, as the only place where an error can be introduced is in the preparation of a quantum state implementing the finite difference coefficients, which is essentially linear in the number of coefficients and subdominant. \textit{Unless otherwise noted, we will abuse the standard notation used above and take $\epsilon$ to be the global truncation error of the finite difference scheme.} 

 We will use the method of linear combination of unitaries (LCU)~\cite{childsHamiltonianSimulationUsing2012} to implement the desired block encodings. LCU begins by assuming that {one has} a matrix operation $A$ that has a decomposition into a sum of $k$ unitary matrices:
\begin{equation}
    A = \sum_{l=0}^{k-1}c_l U_l.
\end{equation}
The goal is to embed the matrix operation given by $A$ into a subspace of a larger matrix $U_A$ which acts as a unitary on the system and ancilla registers. Because of the restriction to unitary dynamics, we require that $||A|| \leq 1$. Therefore we introduce a rescaling factor $\lambda$ that upper bounds the $1-$norm of $A$, which we call the \textit{subnormalization} factor so that $||A/\lambda|| \leq 1$. The block encoding $U_A$ is expressed in matrix form {as}
\begin{equation}
    U_A = \begin{pmatrix}
        A/\lambda & *\\
        * & *
    \end{pmatrix}.
\end{equation}
With LCU, it is straightforward to obtain an upper bound on $||A||$, {since} $\lambda = \sum_{l=0}^{k-1}|c_l|$. Then, scaling the LCU by $1/\lambda$,
\begin{equation}
    \frac{A}{\lambda}= \sum_{l=0}^{k-1}\frac{c_l}{\lambda}U_l \equiv \sum_{l=0}^{k-1}\beta_l U_l,
\end{equation}
we have a matrix that can be block encoded using LCU.

The unitary description of LCU is as follows. We use an ancilla register of $q = \ceil{\log(k)}$ qubits and define the oracle which ``prepares" the LCU coefficients {by}
\begin{equation}
    \textsc{prep}\ket{0}^{\otimes q} \rightarrow \sum_{l=0}^{k-1}\sqrt{\beta_l}\ket{l},
\end{equation}
which {defines} a normalized quantum state and corresponds to the first column of a unitary matrix. Next, we implement the controlled-application of the desired unitaries using the ``select" routine
\begin{equation}
    \textsc{sel}:= \sum_{l=0}^{k-1}\ket{l}\bra{l}\otimes U_l.
\end{equation}
{The block encoding can then} be expressed as the $n+q$ qubit unitary
\begin{equation}    
    U_A := \left(\textsc{prep}^\dagger\otimes I_n\right)\cdot \textsc{sel} \cdot \left(\textsc{prep} \otimes I_n\right).
\end{equation}
The circuit diagram corresponding to this {unitary} transformation is shown in Fig. \ref{fig:LCU}. 

\begin{figure}[ht]
    \centering
    \begin{tikzpicture}
    \hspace{-.3cm}
\node[scale=.98]{
    \begin{quantikz}
        \lstick{$\ket{0}^{\otimes q}$}& \gate{\textsc{prep}}
        \gategroup[2,steps=3,style={dashed,rounded
        corners,fill=blue!20, inner
        xsep=2pt},background,label style={label
        position=below,anchor=north,yshift=-0.2cm}]{{$(\lambda, q)-BE(A)$}}
        &\qw \mathlarger{\mathlarger{\mathlarger{\mathlarger{\oslash}}}} \vqw{1} &\gate{\textsc{prep}^\dagger} &\qw\rstick{$\bra{0}^{\otimes q}$}\\
        \lstick{$\ket{\psi}$ }& \qw  &\gate{U_l} &\qw&\qw\rstick{$\frac{A}{\lambda}\ket{\psi}$} 
    \end{quantikz}
    };
    \end{tikzpicture}
    \caption{Circuit for linear combination of unitaries {(LCU)}. $\textsc{prep}$ prepares the LCU coefficients $\beta_l$ into a quantum state $\sum_{l=0}^{k-1} \sqrt{\beta_l}\ket{l}_q$, and the $\textsc{sel}$ operator applies the unitaries in the LCU in a controlled manner on the state of the ancilla register. The $\oslash$ notation is to be interpreted as a control on all states in the ancilla register. After uncomputing the ancilla register with $\textsc{prep}^\dagger$, the block encoded operator is guaranteed to be applied to the system register upon measuring the ancillary system of $q = \ceil{\log(k)}$ qubits being measured in the all-zero state, where the notation $\bra{0}^{\otimes q}$ at the end of the ancilla wire refers to postselection on the outcome of the ancilla register.}
    \label{fig:LCU}
\end{figure}

\subsection{Shift operators}
\label{subsec:shift operators}
An important class of unitary matrices that we will use regularly in this work are the \textit{shift operators}. The shift operators are well studied \cite{rieffel2014quantum,campsExplicitQuantumCircuits2022,barencoElementaryGatesQuantum1995}, and we give a construction of the incrementer circuit in appendix \ref{app:circs} for those unfamiliar. 
The unitary shift operator $S \in \mathbb{R}^N$ performs a modular increment of the input state:
\begin{equation}
    S:\ket{i} \mapsto \ket{(i+1)\mod N}.
    \label{eq: Right Shift}
\end{equation}
and for an $n$ qubit state $\ket{i}$, the increment by $j$ operator satisfies
\begin{equation}
    S^{j}:\ket{i}_n \mapsto\ket{(i+j)\mod 2^n}_n; \hspace{.2cm} j \in [2^n].
\end{equation}
These operators are 1-sparse and their matrix representations correspond to periodic bands about the diagonal. As an example, the modular increment operator $S^1$ corresponds to the matrix
\begin{equation}
    S^1 = \begin{pmatrix}
        0 & 0 & \cdots&0 & 1\\
        1 & 0 & \cdots&0 & 0\\
        0 & 1 & \cdots&0 & 0\\
        \vdots &   &\ddots & & \vdots\\
        0 & 0 &\cdots &1 & 0
    \end{pmatrix}.
\end{equation}
In Appendix \ref{app:circs}, we show that these operators can be implemented using $O(n)$ multi-controlled Toffolis, which can be expressed as a quantum circuit using $O(n^2)$ Toffoli gates \cite{barencoElementaryGatesQuantum1995}. Using an additional ancilla qubit, this cost can further be reduced to $O(n)$ using the techniques given in \cite{ConstructingLargeIncrement,gidneyHalvingCostQuantum2018}, which is the scaling we will use in this paper. As the ancilla qubit can be reused for each shift operator we employ in our block encodings, we will ignore this additional global ancilla qubit in our cost estimates to simplify the discussion.

\subsection{Quantum Signal Processing}
Quantum signal processing (QSP) and its variants \cite{lowHamiltonianSimulationQubitization2019,lowOptimalHamiltonianSimulation2017,gilyenQuantumSingularValue2019, motlaghGeneralizedQuantumSignal2023} provides a systematic procedure for implementing a class of polynomial transformations to block encoded matrices. 
For Hermitian matrices $A$, the action of a scalar function $f$ can be determined by the eigendecomposition of $A = UDU^\dagger$, as
\begin{equation}
    f(A) := U \sum_{i}f(\lambda_i)\ket{i}\bra{i}U^\dagger,
    \label{eq:mat func}
\end{equation}
where $\lambda_i$ is the $i$th eigenvalue, and $U\ket{i}$ is the $i$th eigenvector. When $A$ is not Hermitian, a generalized procedure based on the singular value transform \cite{gilyenQuantumSingularValue2019} can be used to implement these functions. To implement a degree $d$ polynomial, QSP  uses $O(d)$ applications of the block encoding and therefore its efficiency depends closely on the rate {at which} a given polynomial approximation converges to the function of interest, {as well as on} the circuit complexity {required} to implement the block encoding. 

{QSP uses} repeated applications of the block encoding circuit {to implement} a polynomial of the block encoded matrix. We assume that $A$ is a Hermitian matrix that has been suitably subnormalized by some factor $\lambda$ so that $\norm{A/\lambda}\leq 1$. QSP exploits the fact that the block encoding $U_A$ can be expressed as a direct sum over one and two-dimensional invariant subspaces which correspond directly to the eigensystem of $A$, without knowledge of the eigendecomposition. For each eigenvector $\ket{x}$ of $A$, there is a $2\times 2$ rotation matrix $O_x$ representing the $2d$ subspace associated with $\ket{x}$. In turn this allows one to \textit{obliviously} apply polynomial functions to the eigenvalues in each subspace, which builds up the polynomial transformation of the block encoded matrix $A/\lambda$ overall. QSP is characterized by the following theorem, which is a rephrasing of Theorem 4 of Ref. \cite{lowOptimalHamiltonianSimulation2017}.

\begin{lem}[Quantum Signal Processing (Theorem 7.21 \cite{Lin2022})]
\label{lem:qsp}
    There exists a set of phase factors $\boldsymbol{\Phi} := (\phi_0, \ldots, \phi_{d}) \in \mathbb{R}^{d+1}$ such that
    \begin{equation}
    \begin{aligned}
        U_\phi(x) &= e^{i \phi_0 Z}\prod_{j=1}^d[O(x)e^{i\phi_j Z}]\\
        &=
        \begin{pmatrix}
            P(x) & - Q(x)\sqrt{1-x^2}\\
            Q^*(x)\sqrt{1-x^2} & P^*(x)
        \end{pmatrix},
    \end{aligned}
    \end{equation}
    where
    \begin{equation}
        O(x) = \begin{pmatrix}
            x & -\sqrt{1-x^2}\\
            \sqrt{1-x^2}& x
        \end{pmatrix},
    \end{equation}
    if and only if $P,Q \in \mathbb{C}[x]$ satisfy
    \begin{enumerate}
        \item $deg(P) \leq d, deg(Q) \leq d-1$,
        \item $P$ has parity $d \mod 2$ and $Q$ has parity $d-1 \mod 2$, and
        \item $|P(x)|^2 + (1-x^2)|Q(x)|^2 =1 \hspace{.2cm}\forall x \in[-1,1]$.
    \end{enumerate}
\end{lem}
Given a scalar function $f(x)$, we can use QSP to approximately implement $f(A/\lambda)$ by using the block encoding of $A$, a polynomial approximation to the desired scalar function, and a set of phase factors corresponding to the approximating polynomial. If $f$ is not of definite parity, we can use the technique of linear combination of block encodings to obtain a polynomial approximation to $f$ that is also of indefinite parity. This corresponds to implementing QSP for the even and odd parts respectively and then using an additional ancilla qubit to construct their linear combination. The last requirement, that $P(x) + (1-x^2)|Q(x)| = 1$ for every $x$, is the most restrictive and requires normalization of the desired function, which can be severe in some cases. However, the polynomial $P$ can be specified independent of the polynomial $Q$ so long as $|P(x)|\leq 1 \hspace{.2cm} \forall x \in[-1,1]$ and $P(x)\in \mathbb{R}[x]$ (see Corollary 5 of Ref. \cite{gilyenQuantumSingularValue2019}), which is {indeed} the case for the {algorithms} presented in this work.

{We note that} there are efficient algorithms for computing the phase factors {of QSP} for very high degree polynomials.  Algorithms for finding phase factors for polynomials of degree $d = O(10^7)$ have been reported in the literature (see e.g. \cite{motlaghGeneralizedQuantumSignal2023, dongEfficientPhasefactorEvaluation2021}) and are surprisingly numerically stable even to such high degree. Given that it is straightforward to obtain polynomial approximations to scalar functions and to obtain the corresponding phase factors for QSP, our explicit block encodings combined with {the QSP framework immediately provides a nearly fully explicit circuit description of the block encodings in this work.}

\section{\label{sec:Laplacian} Block encoding of Laplacian with boundary conditions}
In this section, we will provide constructions for the block encoding of the Laplacian. Throughout this section, we will assume that the overall scaling constant of $1/h^2$ has been multiplied through, and is enforced in the state preparation portion of the algorithm. Although this greatly improves the subnormalization factor, it may induce additional complexity in the preparation of the initial state.  We will first briefly review the block encoding of the Laplacian with periodic boundary conditions using a linear combination of shift operators in \ref{subsec:periodic}. We then show how to obtain a block encoding of the periodic operator in $d$ dimensions. We then turn to the implementation of non-periodic boundary conditions in \ref{subsec:BCs simple} and show how the construction for the periodic operator can be efficiently updated to represent Dirichlet, Neumann, and Robin boundary value problems on generalized and non-simply connected rectangular domains in $d$-dimensions using a low-order central difference scheme. We then show how the Dirichlet and Neumann boundary value problems can be discretized using a high-precision scheme using a domain extension in \ref{subsec:periodic extension}. Finally, in \ref{subsec:irregular domains} we address the block encoding of differential operators with boundary data on irregular domains that cannot be described as a simple Cartesian product in $d$ dimensions.  
\subsection{Periodic Boundary}
\label{subsec:periodic}
For the Laplacian with a periodic boundary condition, approximated with a 3-point central-difference stencil with $d=1$, the discrete Laplacian matrix will take the form
\begin{equation}
L=
    \begin{pmatrix}
        -2 & 1 & 0 & \cdots & 0& 1\\
        1 & -2 & 1 & \cdots & 0& 0\\
        0 & 1 & -2 & \cdots & 0& 0\\
        \vdots & \ddots & \ddots &\ddots &\vdots&\vdots\\
        0 & 0 & 0 &\ldots &-2 &  1\\
        1 & 0 & 0 &\ldots &1  & -2
    \end{pmatrix}.
    \label{eq:Periodic 1D Lap}
\end{equation}
For $d=1$, the Laplacian term $\Delta$ with periodic boundary and 3-point finite difference approximation can be expressed as the linear combination of shift matrices introduced in \ref{subsec:shift operators}. It is straightforward to show that the 3-point stencil finite difference scheme applied to the continuum Laplacian operator $\mathcal{L}$ with periodic boundary conditions can be expressed as the linear combination
\begin{equation}
    {L} = S^{-1} -2I + S^1.
    \label{eq:1d LCU}
\end{equation} 
In general for a $2a+1$ point finite difference stencil, the periodic Laplacian can be expressed with $2a$ shift matrices, namely $\{S^{-1}, S^{-2}, \ldots, S^{-a}\}$, $\{S^1, S^2, \ldots, S^a\}$ and an identity matrix. The block encoding for the high-order scheme for approximating the Laplacian with periodic boundary conditions is characterized by the following theorem.

\begin{thm}[Block Encoding Periodic Laplacian]
\label{thm:PeriodicLaplacian}
Using the shift operators and the identity matrix, an $p=2a+1$ point finite difference stencil with $N$ discretization points can be block-encoded with a linear combination of $p$ shift matrices, each of which can be implemented with $O(\log(N))$ Toffoli or simpler gates for a total of $\widetilde{O}(p\log(N))$ Toffoli or simpler gates.
\begin{proof}
    The central finite difference approximation of the second derivative of order $a$ has $2a+1$ coefficients that take on the form given in \eqref{eq:FD coeffs lap}. Using the central difference formula for the second derivative, we may write
    \begin{equation}
        \label{eq:Lap LCU}
        {L} \approx \frac{1}{h^2}\sum_{j=-a}^a r_j S^{j} + O(N^{-2a}),
    \end{equation}
    we immediately obtain a decomposition of the central finite difference approximation of the Laplacian into a linear combination of unitary matrices. We choose $\lambda$ equal to the 1-norm of the coefficients, namely
    \begin{equation}
        \lambda = h^2\sum_{j=-a}^a |r_j|. 
    \end{equation}
Furthermore, by use of the inequality $(a-j)!(a+j)!\geq (a!)^2$, for $0 \leq j \leq a$ and taking the limit as $a \rightarrow \infty$, the sum is given by a Basel problem and consequently we can bound the coefficients in the sum by a constant
\begin{equation}
    \frac{\lambda}{h^2}= \sum_{j=-a}^a |r_j| \leq \frac{2\pi^2}{3},
    \label{eq: alpha upper bound}
\end{equation}
{as shown} in Lemma 6 of Ref. \cite{kivlichanBoundingCostsQuantum2017}. Since the construction of Ref. \cite{ConstructingLargeIncrement} promises that each of the shift operators can be implemented with $O(n)$ Toffoli and T gates with a single ancilla qubit that can be recycled, the Laplacian with periodic boundary conditions can be block encoded to precision $\epsilon \sim N^{-p+1}$ with gate cost $\widetilde{O}(np)$ Toffoli or simpler gates gates, where logarithmic factors arise from the compilation of the rotation angles in the LCU into Clifford + $T$ gates, employing $\ceil{\log(p)}+1$ ancilla qubits.
\end{proof}
\end{thm}
This LCU of the periodic Laplacian is the core of the block encodings we will construct for the Laplacian with various boundary conditions and in arbitrary dimension. The remainder of this section will show how this LCU can be efficiently updated using $O(1)$ additional quantum resources to block encode many other kinds of boundary conditions for $d>1$ dimension.

\subsubsection{d dimensional periodic boundary}
We assume that the $d$-dimensional boundary corresponds to a $d-$dimensional torus, so that the periodic boundary condition can be specified for each dimension independently. In this case, the corresponding operator will satisfy the local tensor product structure discussed in Sec. \ref{subsub:quantum encoding}, so that the $d-$dimensional Laplacian can be constructed as the direct sum of $1d$ Laplacian matrices. With this structure, we can obtain an $(O(d),\log(d) + \log(p),\frac{d}{N^{p-1}})-$ block encoding of the approximate periodic Laplacian operator.

\begin{figure}[h]
    \centering
    \begin{quantikz}
    \lstick{$\ket{0}^{\otimes \ceil{\log(d\eta)}}$}&\qwbundle{dims} 
        &&\gate{D_{d\eta}}
        \gategroup[3,steps=3,style={dashed,rounded
        corners,fill=blue!20, inner
        xsep=2pt},background,label style={label
        position=below,anchor=north,yshift=-0.2cm}]{{$(\lambda \eta d, \log(p)+\log(\eta d))$-BE$(\bs{L})$}}&\oslash\vqw{1}&\gate{D_{d\eta}^\dagger}&\qw\\
        \lstick{$\ket{0}^{\otimes\ceil{ \log(p)}}$}&\qwbundle{stencil}&\hphantomgate{} &\gate{\textsc{prep}_r}&\oslash\vqw{1}&\gate{\textsc{prep}_r^\dagger}&\qw\\
        \lstick{$\ket{b}$}&\qwbundle{sys}&&&\gate{\text{sign}(r_j)S_l^j}&\qw&
    \end{quantikz}
    \caption{Quantum circuit for LCU of Laplacian  $\bs{L}$ with periodic boundary conditions in $d$ spatial dimensions using a stencil of arbitrary order $p = 2a+1$ and for $\eta$ particles. $\textsc{prep}$ is the unitary which prepares the unsigned finite difference coefficients, $\text{sign}(r_j)r_j$ on the \textit{coeff} register of size $O(\log(p))$ and $D_d$ is the unitary which prepares the equal superposition over $d$ states on the \textit{dims} register of size $O(\log(\eta d))$. The system register \textit{sys} holds $\eta d$ registers of $n$ qubits representing a $N^{\eta d}$ dimensional state $\ket{x}$. $S^j_l$ refers to the modular increment-by-$j$ operator in register $l$ and $\text{sign}(r_j)S^j$, multiplies the overall unitary by a phase $\pm 1$. We use the $\oslash$ convention to mean the application of a unitary controlled on any $\ket{j}$ in the \textit{stencil} register for every $j \in [\eta]$. Similarly, for the \textit{dim} register, $\oslash$ controls on $\ket{l}$ for every $l \in [d]$.}
    \label{fig:d dim periodic circuit}
\end{figure}

We can construct the block encoding of this matrix as follows. Let 
\begin{equation}
    \textsc{prep}_r\ket{0} = \sum_{j}\sqrt{\frac{|r_j|}{\lambda}}\ket{j},
\end{equation}
where $\lambda = \sum_{j}|r_j|\leq \frac{2\pi^2}{3}$. In practice, $p = O(1)$ and so the sign information of the coefficients is assumed to be easily {accessible} and can efficiently be included in the unitary. Therefore, in the state preparation we take the $r_j$'s to be non-negative and will incorporate the sign information into the unitaries by multiplying them by an appropriate phase. Now, we let 
\begin{equation}
    \textsc{sel} = \sum_{i=-a}^a \text{sign}(r_i)S^i\otimes \ket{i}\bra{i}.
\end{equation}
Then, we introduce the notation $S^j_l$ to refer to the increment by $j$ operator applied to register $l$, acting trivially on the other registers. In circuit form, this corresponds to the following equivalence:
\begin{equation*}
    \begin{quantikz}
    \lstick{$\ket{x}_0$}&\qwbundle{n_0}&\gate[5]{S^j_l}&&\\
    \vdots\\
    \lstick{$\ket{x}_l$}&\qwbundle{n_l}&&&\\
    \vdots\\
    \lstick{$\ket{x}_{d-1}$}&\qwbundle{n_{d-1}}&&&
    \end{quantikz}
    \hspace{.5cm}\equiv
    \begin{quantikz}
    \lstick{$\ket{x}_0$}&\qwbundle{n_0}&&&\\
    \vdots\\
    \lstick{$\ket{x}_l$}&\qwbundle{n_l}&\gate{S^j}&&\rstick{$\ket{x+j\mod{N}}_l$}\\
    \vdots\\
    \lstick{$\ket{x}_{d-1}$}&\qwbundle{n_{d-1}}&&&
    \end{quantikz}.
\end{equation*}
Now, we observe that {the $d$-dimensional system can be represented by the following LCU:}
\begin{align*}
    &\frac{1}{\lambda d}\sum_{i\in[d]}\sum_{j=-a}^{a}r_j S^j_i\\
    &= \frac{1}{\lambda d}\sum_{i\in[d]}\sum_{j=-a}^{a}r_j I^{\otimes i}\otimes S^j \otimes I^{\otimes (d-i-1)}\\
    &=\frac{1}{\lambda d}\sum_{i\in[d]}\mathcal{L}_i^{(d)}\\
    &=\frac{1}{\lambda d} \mathcal{L},
\end{align*}
where $\mathcal{L}_i^{(d)} = \sum_{j=-a}^{a}r_jS^j_i$ is the one-dimensional periodic Laplacian on register $i$, acting as the identity on all other registers. Note that this block encoding does not use controlled applications of the block encodings of the $d=1$ dimensional operator.

By preparing the uniform superposition over $d$ states, providing indices to control which register we apply the shifts on, we can implement the $d$ dimensional Laplacian matrix with comparable efficiency to the $1d$ case, {because} the number of additional ancilla qubits is only an additional $O(\log(d))$ larger than the $1d$ case, the circuit depth is $O(d p n)$ and the subnormalization is $d\lambda$. The quantum circuit implementing this block encoding is shown in Fig. \ref{fig:d dim periodic circuit}. The generalization to $\eta$ particles in $d$ dimensions follows the same pattern as the generalization from $1$ to $d$-dimensions.

\subsection{Dirichlet, Neumann, and Robin Boundary conditions}
\label{subsec:BCs simple}
In this section, we show how Dirichlet, Neumann, and Robin boundary conditions can be implemented using the low order scheme to enforce common boundary value problems that we introduced in section \ref{sec:Background}. These block encodings of the Laplacian with non-periodic boundary conditions do not enjoy the same convergence rates as the higher-order finite difference scheme used in the periodic operator. This is because the evaluation of the differential operator near the boundary restricts the order of the finite difference stencil that can be used, which  
lowers the global convergence rate. Nevertheless, these schemes are standard in classical {methods for the numerical solution} of differential equations, and so we write the circuits for these block encodings explicitly {in this section}, which may be useful for simple applications. Later, we will show how Neumann and Dirichlet conditions can be {efficiently} implemented with high-precision using an extended domain and forcing function.

In $d=1$ the boundary conditions only affect the rows of the corresponding matrix which describe the behavior of the solution on the boundary. When the boundary value problem is on the extremal points of an interval, upon discretization, the boundary conditions change the matrix elements of the periodic Laplacian on only the first and last row. Here we demonstrate how to implement this change for the simple case of a Dirichlet boundary condition that is specified on both endpoints of the domain. Recall from \eqref{eq:dir mat} that up to a global scaling factor, the 3-point stencil for the Laplacian with Dirichlet boundary corresponds to the tridiagonal matrix 
\begin{equation*}
{L}_D = 
       \begin{pmatrix}
        -2 & 1 & 0 &\cdots & 0\\
        1 & -2 & -1&\cdots & 0\\
        \vdots& \ddots &\ddots &\ddots &\vdots\\
        0&\cdots & 1 &-2 &1\\
        0&0 &\cdots & 1 & -2
    \end{pmatrix}.
\end{equation*}

We can construct this matrix with a similar LCU to the periodic case, but with the addition of reflection operations. This is based on the simple observation that a sum or difference of a reflection and an identity matrix forms the projector onto the reflected subspace or its complement. Notice that we may express $L_D$ as,
\begin{equation*}
    {L}_D = -2I + \Pi_{0}^{\perp}S^{-1} + \Pi_{N-1}^{\perp}S^{1}.
\end{equation*}
The projectors $\Pi_{0}^{\perp} = I - \ket{0}\bra{0}$ and $\Pi_{N-1}^{\perp}= I-\ket{N-1}\bra{N-1}$ can be expanded using a linear combination of the reflections $R_0 = I - 2 \Pi_{0}$, $R_{N-1} = I - 2\Pi_{N-1} $ and the identity as, 
\begin{equation}
    {L}_D = -2I + \frac{1}{2}\left(R_{0} + I\right)S^{-1} + \frac{1}{2}\left(R_{N-1} + I\right)S^{1}.
\end{equation}
This LCU can be immediately implemented as an $(\lambda_D = 4, m=3,\epsilon=O(N^{-2}))$ block encoding of the Laplacian with Dirichlet boundary conditions on $N$ grid points. The quantum circuit implementing this block encoding is given in {Fig.} ~\ref{fig:dirichlet lcu circuit}, and involves only a constant number of additional gates and only a single additional ancilla qubit.

The Neumann and Robin boundary value problems are similar. For the on dimensional case, the matrices we need to block encode take the form
\begin{equation}
L_N = 
        \begin{pmatrix}
        -1 & 1 & 0 &\cdots & 0\\
        1 & -2 & 1&\cdots & 0\\
        \vdots& \ddots &\ddots &\ddots &\vdots\\
        0&\cdots & 1 &-2 &1\\
        0&0 &\cdots & 1 & -1
    \end{pmatrix},
    \hspace{.2cm}
L_R = 
    \begin{pmatrix}
        -1-ah/b & 1 & 0 &\cdots & 0\\
        1 & -2 & 1&\cdots & 0\\
        \vdots& \ddots &\ddots &\ddots &\vdots\\
        0&\cdots & 1 &-2 &1\\
        0&0 &\cdots & 1 & -1+ch/d
    \end{pmatrix},
\end{equation}
{Here} we need to adjust three matrix elements on the boundary rows. Using a similar approach with a linear combination of shifts and shifts multiplying reflections, we can also construct the matrix corresponding to the Neumann boundary value problem. For the Neumann case, corresponding LCU can be expressed 
\begin{equation}
    {L}_N = \frac{-3}{2}I -\frac{1}{2}R_{0}R_{N-1} + \frac{1}{2}\left(R_{0}S^{-1} + S^{-1}\right) + \frac{1}{2}\left(R_{N-1}S^1 + S^1\right).
    \label{eq:NeumLCU}
\end{equation}
This constructs a $\left(\lambda_N = 4, m = 3, \epsilon = O(N^{-1})\right)$ block encoding of the Neumann boundary value problem.
The Robin condition can {now} be encoded by adjusting the the first and last diagonal entries of the block encoding of the Neumann matrix. This can be accomplished by adding the following LCU terms
\begin{equation*}
    -\frac{a h}{2b}\left(I - R_{N-1}\right) + \frac{c h}{2b}\left(R_{0} - I\right).
\end{equation*}
Therefore, the Robin LCU can be expressed
\begin{equation}
    {L}_R = \frac{(c-a) h-3b}{2b}I+\frac{a h}{2b} R_{N-1} + \frac{c h}{2b}R_{0}  -\frac{1}{2}R_{0}R_{N-1} + \frac{1}{2}\left(R_{0}S^{-1} + S^{-1}\right) + \frac{1}{2}\left(R_{N-1}S^1 + S^1\right).
    \label{eq:rob BCs lcu}
\end{equation}
This results in an $(\lambda_R \leq 4 +\frac{c+a}{Nb}, m=4)$ block encoding of the linear system corresponding to Robin boundary condition.  
In many cases with Neumann or Robin conditions, the stability of the corresponding linear system may have a non-trivial dependence on $N$, typically $O(N)$, which will lower the global truncation error of the scheme to $\epsilon\sim O(N^{-1})$. Furthermore, the stability of the system depends on the coefficients $a$ and $b$ as well as the associated boundary values. Additionally, since there are terms with different powers of $N$ in the summation, the factor of $N^{2}$ cannot be viewed as a global scaling factor and must be incorporated into the coefficients. {Although}, the subnormalization factor for the Neumann/Robin case {becomes} is similar to the Neumann case $\lambda_R \leq \lambda_N + O(1/N)$, the preparation of the LCU coefficients is more involved.

\begin{figure}[h!]
    \centering
  \includegraphics[width=.45\linewidth]{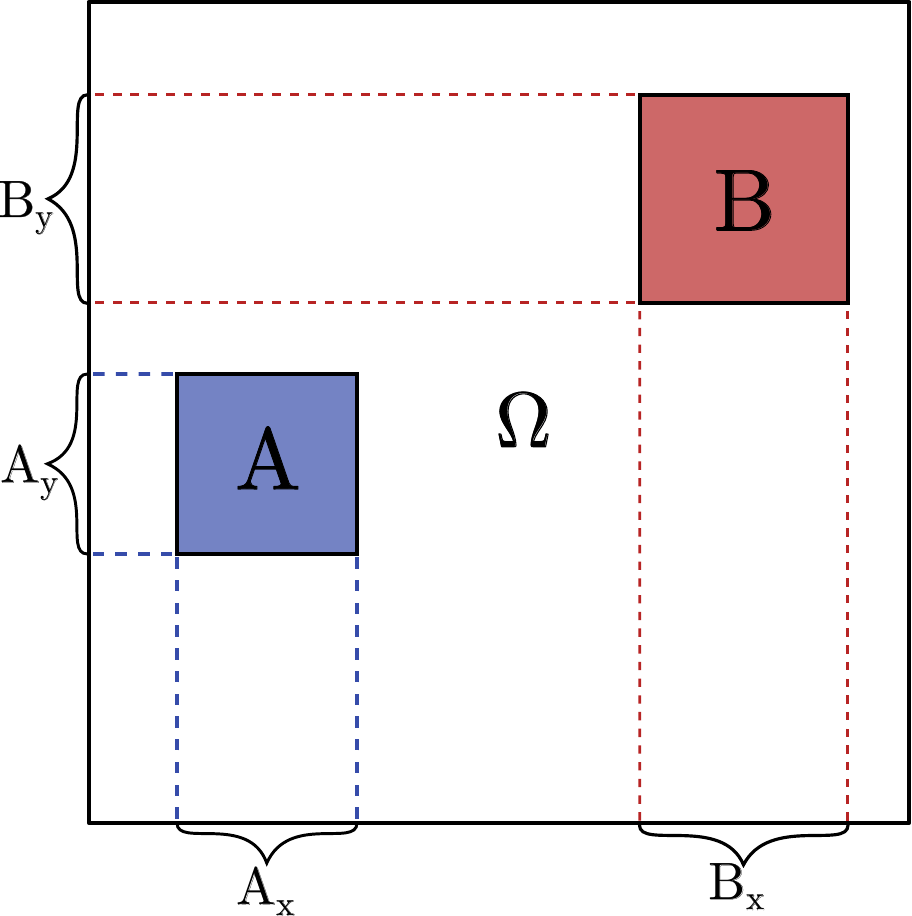}
    \caption{{An example} of a non-simply connected domain. For this portion of the discussion, we assume that the boundary conditions are given in each dimension independently so that the projection of the shapes of the boundary onto the principal axes is sufficient to completely characterize the hole. 
    We solve the Dirichlet problem corresponding to the holes $A$ and $B$ by solving the $d=1$ Dirichlet problems on $A_x \cup B_x$ and $A_y\cup B_y$, assuming the solution takes a constant value on the interior of the holes. Because the underlying operator is still separable, knowledge of the projections of $A$ and $B$ and the block encodings for the Dirichlet problem for $d=1$ are sufficient to block encode the Laplacian on this non-simply connected domain.}
    \label{fig:non simple domain}
\end{figure}

This scheme can be straightforwardly extended to dimension $d>1$ using a similar circuit to that in Fig. (\ref{fig:d dim periodic circuit}), {as} long as the underlying boundary conditions are separable. In addition, this approach also allows for the evaluation of boundary conditions{ imposed}, for example, on the interior of the grid. To illustrate this, consider the following {example}. Let $[0,L]^d$ be the domain and $A = A_0 \cup A_1 \cdots \cup \cdots A_{d-1}$ be a set of points corresponding to the boundary in each dimension so that Cartesian product $(A_0 \times A_1 \times \cdots \times A_{d-1}) \subset [0,L]^d$ corresponds to a {generalized} rectangle in $d$-dimensions. Because each of the faces of the rectangle are orthogonal, the projection of the rectangle onto the axes in each dimension is sufficient to describe the shape. Exploiting the separable structure, in each dimension $i$, we can block encode the Dirichlet boundary value problem with the LCU
\begin{equation}
    \mathcal{L}_i^{(d)} = 2I -\frac{1}{2}\left(S^{-1}R_{A_{i}+1} + S^{-1} + S^{1} + S^{1}R_{A_{i}-1}\right),
    \label{eq:dir NS}
\end{equation}
{with} $A_{i}+j \equiv \{j+k \mod N :\forall k \in A_{i}\}$ and $R_{A} = I - 2\sum_{j\in A} \ket{j}\bra{j}$. This LCU produces a matrix that approximates the Laplacian on $[0,\min(A_{i}))\cup (\max(A_{i}), L]$ and is proportional to the identity matrix on grid points supported in $[\min(A_{i}),\max(A_{i})]$, so that the prescribed value on the interior can be enforced by simply placing the boundary values in the elements of right hand side vector that correspond to the grid points on the boundary. The Neumann boundary value problem can be block encoded similarly.

\begin{figure}[h!]
    \centering

    \begin{quantikz}
        \lstick{$\ket{0}$}&\qw&&\gate{H} \gategroup[4,steps=5,style={dashed,rounded
        corners,fill=blue!20, inner
        xsep=2pt},background,label style={label
        position=below,anchor=north,yshift=-0.2cm}]{{$(4 d, \log(d)+3)$-BE$({\mathcal{L}^{(BC)})}$}}
        &\ctrl{1}&\qw&\octrl{1}&\gate{H}&\qw \\
        \lstick{$\ket{0}$}&\qwbundle{dim}&\hphantomgate{} &\gate{D_d}&\oslash\vqw{1}&&\oslash\vqw{1}&\gate{D_d^\dagger}&\qw\\
        \lstick{$\ket{0}$}&\qwbundle{stencil}& &\gate{\textsc{prep}_r}&\oslash\vqw{1}&&\oslash\vqw{1}&\gate{\textsc{prep}_r^\dagger}&\qw\\
        \lstick{$\ket{x}$}&\qwbundle{sys}&\hphantomgate{} & &\gate{S_l^jR_{A_l-j}}&&\gate{S_l^j}&\qw&
    \end{quantikz}
    \caption{Quantum circuit for LCU of second derivative operator with boundary conditions on a simply or non-simply connected domain in $d$-dimensions using low-order stencil $\mathcal{L}^{(BC)}$. The subnormalization $4d$ results from $\lambda = 4$ for the 3 point stencil, which we then combine $d$ times. Here we use the notation $\textsc{prep}$ as the unitary which prepares the finite difference coefficients on the \textit{stencil} register and $D_d$ as the unitary which prepares the equal superposition over $d$ states on the \textit{dim} register where $S^j_l$ refers to the modular increment-by-$j$ operator in direction $l$. We use the $\oslash$ convention to mean controlled over every state. Appending an additional register of $\ceil{\log(\eta)}$ qubits and selecting over shift operators both in dimension and particle indices will implement the $\eta$-particle Laplacian over $d$-dimensions.}
    \label{fig:dirichlet lcu circuit}
\end{figure}

Because the LCU involves a linear combination of shift operators and shift operators multiplying reflection matrices, which have known circuit decompositions into one and two qubit gates, we can immediately obtain a decomposition of the LCU into one two qubit gates. In particular, the shift operators have been discussed many times in the literature, and an optimized implementation is provided in Ref. \cite{ConstructingLargeIncrement}, using an additional ancilla qubit and $n-1$ Toffoli gates. The reflection matrices can be implemented by n-qubit control-Z gates and single-qubit Pauli gates, which can also be implemented by $O(n)$ Toffoli gates. Therefore, if the boundary conditions are set on the exterior grid points and satisfy the above assumptions, we can use the periodic extension to enforce the Neumann and Dirichlet to high precision in $d$-dimensions. In particular, using the 3 point scheme in $d$-dimensions the expected number of shift and reflection operators needed is $O(2d)$ to block encode the Neumann and Dirichlet matrices as opposed to $O(d)$ needed for the periodic operator. Since each shift and reflection operation requires at most $O(n)$ Toffoli gates to implement, the total Toffoli count for the periodic case is $O(dn)$ and the case with Dirichlet and Neumann boundary is $O(2dn)$. 

The total number of Toffoli or simpler gates is $O\left(dn\right)$, with $\log(d)+1$ ancilla qubits {required} to block encode the $d$-dimensional Laplacian with periodic boundary conditions on a rectangular grid.  Replacing the periodic Laplacian with the Dirichlet block encoding will {provide} a block encoding of the $d$-dimensional Dirichlet problem.  Doing this for $\eta$ particles, we can simply append another register of size $\log(\eta)$ and proceed similarly. This will in turn require $O\left(dn\eta\right)$ Toffoli or simpler gates to implement and $O(\log(dp\eta))$ ancilla qubits.

{We note that} the above schemes for Dirichlet and Neumann/Robin {boundary} conditions are generically limited to second, $O(N^{-2})$, and first order, $O(N^{-1})$, accuracy {with respect to the discretization,} respectively. We can improve these orders of accuracy by using a high-order scheme, but this is complicated by the ambiguity of how to evaluate the high-order scheme near the boundary. The introduction of ghost points is a common technique for handling Neumann boundary value problems, but it is unclear how to extend the ghost point method to arbitrary order. However, if we can uniquely and smoothly extend the PDE onto a larger domain, then we can use a higher order method to evaluate the boundary value problem. With some restriction on the domain and boundary conditions, this can be accomplished, at least in theory, {using} the method of images. Although the method of images also holds for the Robin boundary condition, it is more complicated than the Dirichlet and Neumann boundary conditions, so we will focus on their implementation. We {now} show how this can be efficiently implemented in practice on a quantum computer {using a method of periodic extensions for the Dirichlet and Neumann case.}

\subsection{Neumann and Dirichlet boundary conditions via periodic extensions}
\label{subsec:periodic extension}
One practical way to implement the method of images over a finite interval, say $[0,1]$, is to use the even or odd extensions of the forcing function $f$ onto $[-1,1]$ and {then} discretize the corresponding differential equation over the enlarged interval. Although this doubles the number of grid points, on a quantum computer this can be implemented with {the introduction of just a }single additional qubit. This idea was first explored as part of \cite{highPrecision}, but left to future work its detailed implementation. We provide {one set of the required} implementation details here. For this section, we will assume that the boundary value problem is homogeneous and that the corresponding extended forcing function $\widetilde{f}$ is smooth across the boundary.

From a geometric perspective, this can be viewed as embedding some system with specified non-periodic boundary conditions within a larger system {having} a periodic boundary. The degrees of freedom in the extended portion of the domain are constructed to produce the correct solution on the original domain of definition. This extended system allows the use of central difference schemes, even on boundary nodes, to obtain approximations at the boundary to the same precision as {in} the interior. Furthermore, constructing such an extension is straightforward {since} only the original operator, the domain of definition, and the prescribed boundary conditions are sufficient to uniquely extend the operator onto a new domain and {to} construct the appropriate periodic extension of the forcing function $f$. {To proceed, we first} define even and odd periodic extensions of the forcing function on the right hand side $f$.
\begin{defn}[Even and odd periodic extensions]
    Given a function $f$ defined over some compact interval, say $[0,L]$, $f:[0,L] \rightarrow \mathbb{R}$, we can define the even periodic extension to the interval $[-L, L]$ $\widetilde{f}_e :[-L,L] \rightarrow \mathbb{R}$
    \begin{equation}
        \widetilde{f}_e(x) = \begin{cases}
            f(-x) & x \in [-L, 0)\\
            f(x) &x \in [0,L)
        \end{cases}
    \end{equation}
    so that $f(-L) = f(L), f(0) = f(2L)$. {Similarly, we define the} related odd periodic extension $\widetilde{f}_o:[-L,L]\rightarrow \mathbb{R}$ {as} 
    \begin{equation}
        \widetilde{f}_o(x) = \begin{cases}
            -f(-x) & x \in [-L, 0)\\
            0 & x = 0 \\
            f(x) &x \in [0,L)
        \end{cases}
    \end{equation}
    so that $\widetilde{f}_o(-x) = -\widetilde{f}_o(x)$ is an odd function. In addition, we assume that $\widetilde{f}$ is smooth in $[-L,L]$, meaning that for any $p > 3$, and $x \in [-L,L]$, $\max_{x}\vert f^{(p)}(x)\vert < C$, for some constant $C$.
\end{defn}
The assumption of smoothness over the extended domain is critical to obtaining theoretical guarantees of higher order accuracy at the boundary to achieve an higher-order globally accurate scheme. If such smoothness is not automatically guaranteed by the original problem definition, one can mollify the forcing function through a convolution with a smooth bump function centered at the boundary point to smooth connect the even and odd extensions of the problem. Then we simply solve the equation with the smoothed version of the periodically extended forcing function, accepting the small controllable error introduced from modifying the original function on the domain. Intuitively, by extending the domain of the original function and {of the} differential operator, the stencil will involve points extending across the boundary, causing interference between the solution on the original domain and that on the extended domain. The odd periodic extension enforces a kind of destructive interference of the solution near the boundary, while the even periodic extension enforces the destructive interference of the \textit{derivatives} of the solution near the boundary. This behavior is related to why the Laplacian with Dirichlet and Neumann {boundary value problems (BVPs)} are diagonalized by the discrete sine and cosine transformations respectively. Extending this approach to Robin boundary conditions is less straightforward. In principle it should be possible, using, for example, the approach found in \cite{rosalesRobinBoundaryConditions}, but it is unclear if this method can be efficiently implemented using a straightforward adjustment of the scheme we present for Dirichlet and Neumann BVPs.

For a nonhomogenous Dirichlet problem
\begin{align*}
    \Delta [u](x) &= f(x) \text{;   }\hspace{.2cm} 0 < x < L,\\
    u(0) &= f(0),\\
    u(L) &= f(L),
\end{align*}
we must first homogenize the system to apply the method of images.
To homogenize the boundary value problem, we first introduce a new solution $V = u(x) - \frac{f(L)-f(0)}{L}x - f(0)$ to homogenize the original boundary value problem so that 
\begin{align*}
    \Delta[V] &=  f(x) \text{;   }\hspace{.2cm} 0 < x < L,\\
    V(0) &= 0,\\
    V(L) &= 0.
\end{align*}
Now, we consider the odd periodic extension of the function $f$ about $x=0$ and $x=L$
\begin{equation}
    \widetilde{f}_o(x) = \begin{cases}
        -f(-x) & -L < x < 0\\
         0 &  x = 0\\
        f(x) & 0<x < L\\
        0 & x = L.
    \end{cases}
\end{equation}
In general, since $f(0) \neq f(L) \neq 0$, we may have the even or odd extension may be discontinuous or have discontinuous derivatives across the boundary that would lower the smoothness of the forcing function at the boundary, destroying the smoothness of the solution as promised by lemma \ref{lem:ellip regularity} and thereby the global convergence rate of the high order scheme. In the case where $f$ decays sufficiently quickly so that we need only to perform the periodic extension about a single boundary point, we may convolving $f$ with a smooth bump function $\varphi_k(x) = e^{\frac{1}{(kx)^2-1}}$, for $k>1$, to mollify the forcing function across the boundary. We then choose $k$ large enough that we approximate the periodically extended function within the desired precision and solve the problem on the mollified forcing function given by,
\begin{equation}
    \hat{f}(x) := \frac{1}{\mathcal{N}}\int_{-1}^1  \varphi_k(x-y)f(y) dy,
    \label{eq:mollified f}
\end{equation}
where $\mathcal{N} = \int_{-1}^1 \varphi_k(x)dx$ is a normalizing constant.
For this scheme to be efficient, we require that $k=O(1)$, as the $n$th derivatives of the bump function go as $O(k^n)$. If $k$ is large, one may need to balance of the order of the scheme with the smoothness of $f$ across the boundary. However, we do not expect for $k$ to be so large in practice.  In Fig. \ref{fig:mollified extension}, we demonstrate a particularly challenging example of periodically extending the exponential function across the origin. In the figure, we show close qualitative accuracy for mollifying constant $k = 5$ in the even extension case and $k = 10$ in the odd extension case. For the rest of this section, we will assume that smoothness across the boundary is guaranteed by the original problem since the mollification process can be efficiently performed classically if needed.

\begin{figure}
    \centering
    \begin{minipage}{0.5\textwidth}
        \centering
  \includegraphics[width=\linewidth]{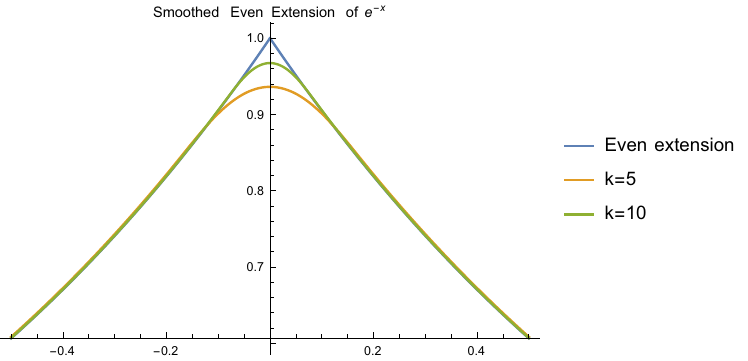}
    \end{minipage}\hfill
    \begin{minipage}{0.5\textwidth}
        \centering
        \includegraphics[width=\linewidth]{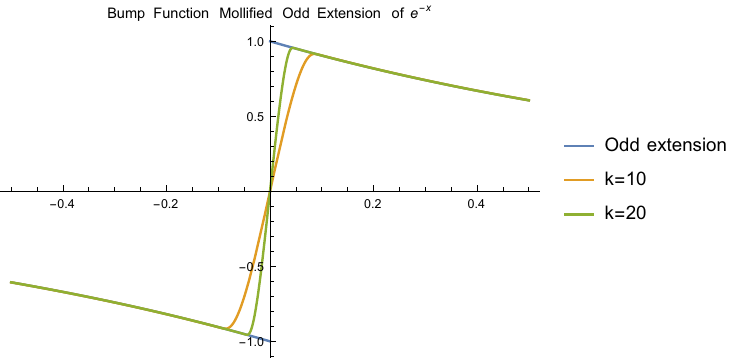}
    \end{minipage}
    \caption{Smoothed even and odd extensions of an exponential forcing function $f(x) = e^{-x}$. This case represents a particularly challenging example, as the even extension of $e^{-x}$ has a cusp at $x=0$ and the odd extension of $e^{-x}$ is discontinuous at $x=0$. For the even extension, we convolve $\widetilde{f}_e(x)$ with $e^{\frac{1}{(kx)^2 -1}}$ for $k= 5$ and $k=10$ in the figure on the left. For the odd extension, we convolve $\widetilde{f}_o(x)$ with the same mollifying function 
    using the mollifying constants $k=10$ and $k=20$ in the figure on the right. The convolution can be accomplished according to Eq. \eqref{eq:mollified f}. These figures demonstrate how a mollified version of the periodic extension can be used to obtain an infinitely differentiable periodic extension of the original forcing function onto the extended domain, at the cost of some error that depends on mollifying constant $k$. This allows us to use a central difference scheme spanning the extended domain, similar to the standard periodic boundary condition case. This greatly improves the smoothness of the solution and the overall convergence rate, correspondingly. We provide numerical examples for the solution of periodically extended domains that are available on Github \cite{KharaziCode}.}
\label{fig:mollified extension}
\end{figure}

Now, consider the differential equation with respect to the linearly shifted solution given above
\begin{align*}
    \Delta[V] &= \widetilde{f}_o(x), \hspace{.2cm} x \in (-L,L)\\
    V(-L) &= V(L)
\end{align*}
with periodic boundary conditions.
Now, we notice that
\begin{align*}
    \Delta[V](-x) = -f(x) = - \Delta[V](x)
\end{align*}
{and} by linearity, we conclude {that}
\begin{align*}
    \Delta[V](0) = \sum_{j=1}^a r_j((V(j)+V(-j)) +r_0V(0) \sim 0 + O(h^{p}),
\end{align*}
which implies $V(0) = O(h^p)$. Therefore, we can use the block encoding of the periodic Laplacian on $2N-2$ grid points applied to the smooth odd extension of $f$ to obtain a block encoding of Dirichlet Laplacian which converges globally as $O(h^{p-1})$. {This can then} be implemented via the LCU \eqref{eq:1d LCU} and correspondingly as an $\left(\lambda = O(1), \ceil{\log(p)}, \epsilon \sim N^{-p+1}\right)$ block encoding of the Dirichlet operator. 

The Neumann case is similar. Consider the periodic boundary value problem with an even function $\widetilde{f}$ about $x=0$ and $x=L$:
\begin{align*}
    \Delta[V] &= \widetilde{f}_e(x) \hspace{.2cm} x\in(-L,L)\\
    V(-L) &= V(L).
\end{align*}
Now suppose $u$ is some particular solution to the above. Then,
\begin{align*}
    u'(0) &\sim \sum_{j=1}^{a} c_j u(-j) - c_j u(j) + O(h^{p})\\
    &= O(h^{p})
\end{align*}
satisfies the required boundary condition to high precision. A similar result can be found at $x=L$. Therefore, once again, this matrix can be block encoded using a high order scheme for the Laplacian satisfying periodic boundary conditions on the extended domain. However, in this case, the subnormalization factor will again depend on $N$, but the factor of $N$ in the global truncation error arising from the stability of the linear system is also only $O(1)$ since the rows with Neumann conditions $R$ satisfy $R\ket{x} = 0$ for every $x$ in the row space of $R$. Therefore the overall global truncation error for the Neumann scheme is similar to that of the Dirichlet case $\epsilon \sim O(N^{-p+1})$. 

Extending to $d$-dimensions with the periodic embedding is also straightforward whenever $f:\mathbb{R}^d \rightarrow \mathbb{R}$ can be expressed as a separable function of its arguments. We simply periodically extend each portion of the domain and {the} function corresponding to each dimension, and then {constitute} the direct sum to obtain the whole operator. Periodic extensions in $d$ dimensions will increase the number of grid points by a factor of $2^d$ so that $N^d \rightarrow N' \sim O(2^d N^d)$, but this only comes at the cost of $\sim d \log(2N)$ qubits, which is essentially equivalent. However, for a fixed $N$, the periodic embedding technique promises $\epsilon \sim N^{-(p-1)}$ whereas the standard schemes only provide $\epsilon \sim N^{-(p-1)/2}$, assuming some high order scheme has been consistently implemented at the boundary, so that one would need at least \textit{quadratically} more grid points to obtain the same accuracy as the periodic embedding scheme. Furthermore, it is much less clear how to extend the above constructions in an efficient {manner to} achieve {the standard} accuracy $O\left(N^{-(p-1)/2}\right)$ in the first place, as whole new coefficients would be needed on the boundary rows for the scheme. Finally, and perhaps most beneficially, the constructed matrix {can be} diagonalized by the quantum Fourier transform, allowing for faster algorithms to solve the underlying linear system or perform dynamical simulations \cite{tongFastInversionPreconditioned2021a,highPrecision}.

\subsection{Block encoding operators on irregular domains}
\label{subsec:irregular domains}
The main technique used to implement the above boundary conditions is to use a linear combination of a unitary with itself, {multiplied by} a reflection about the boundary nodes to zero out {specific} matrix elements. Unfortunately, our technique of constructing explicit block encodings relied on the assumption of separability to efficiently exploit the tensor product structure {and so cannot be immediately applied in less structured domains.} However, a very similar approach can be used for block encoding domains where an explicit tensor product structure is not apparent, if we can efficiently compute a membership oracle for the domain.  We {illustrate this by considering} the following scenario. We have some $d-$dimensional physical domain $\Omega$ that is contained within some $d-$dimensional cube that we call the computational domain, i.e., there is some $L>0$ {such} that  $\Omega \subseteq \Omega_L^d \sim [0,L]^d$. This is similar to classical embedding methods which embed the domain of interest within a subset of a Cartesian domain. These techniques are a starting point for the numerical simulation of curved boundaries \cite{ruuthSimpleEmbeddingMethod2008a,colellaEmbeddedBoundaryAlgorithms2008, xuHighefficiencyDiscretizedImmersed2023,cocoHighOrderFinitedifference2024}, which could be an avenue to explore in future work. We will discretize $\Omega_L^d$ as usual into $N^d$ many grid points, and in addition, we will assume access to an oracle that can compute the following quantities
\begin{equation}
    \textsc{dom}\ket{{x}}\ket{0} = 
    \begin{cases}
    \ket{x}\ket{0} & \text{ if } x \in \text{int}(\Omega)\\
    \ket{x}\ket{1} & \text{ if } x \in \text{bndry}(\Omega)\\
    \ket{x}\ket{2} & \text{ if } x \in \Omega_L^d \backslash \Omega.
    \end{cases}
    \label{eq: dom oracle for BCs}
\end{equation}
We have used the shorthand notation $\ket{x} = \ket{x_0}\ket{x_1}\cdots\ket{x_{d-1}}$
to refer to the $N^d-$dimensional state vector indexing a grid point in $\Omega_L^d$. We further assume that inclusion in the set $\Omega$ can be efficiently computed. Some examples where this can occur are {i)} if the function parameterizing the boundary of $\Omega$ can be expressed as a smooth function or as a small $\widetilde{O}(1)$ family of piecewise smooth curves, or {ii)} if the boundary is given by the convex hull of a simple set, {in which case} such an oracle could be efficiently realized using the method of inequality testing \cite{sandersBlackBoxQuantumState2019}. The resulting block encoding corresponds to an operator with boundary conditions on a ``rasterized" domain. Because this block encoding is still inherently restricted to a Cartesian grid, the $\textsc{dom}$ oracle could induce ``jagged" edges when the domain is a smooth curve like a circle. Although it may be possible to choose $N$ large enough to minimize this effect, this may introduce uncontrolled numerical errors which we will not account for here.

{Now let} $L$ be a linear differential operator defined with periodic boundary conditions in each dimension on $\Omega_c^d$, which is then discretized using standard finite difference methods, and let $U_{L}$ be its $(\lambda, m)$-block encoding. Because the periodic boundary conditions are separable, this block encoding can be obtained by the methods given above. For a generic Robin boundary condition $au(x) + b\nabla u(x)\cdot \mathbf{n}(x) = f$, we may express the operator in difference form as $B = b\sum_{1\leq |j| \leq a}S^jc_j + aI = \mathbf{f}$.   Let $U_{B}$ be the $(\beta, l)$-block encoding of these boundary conditions. We can then implement these operations on the relevant subsets of the domain and sum them together to obtain the encoding of the desired operator on the physically relevant subset. 

\begin{figure}
    \centering
    \centering
    \begin{minipage}{0.45\textwidth}
        \includegraphics[width=\textwidth]{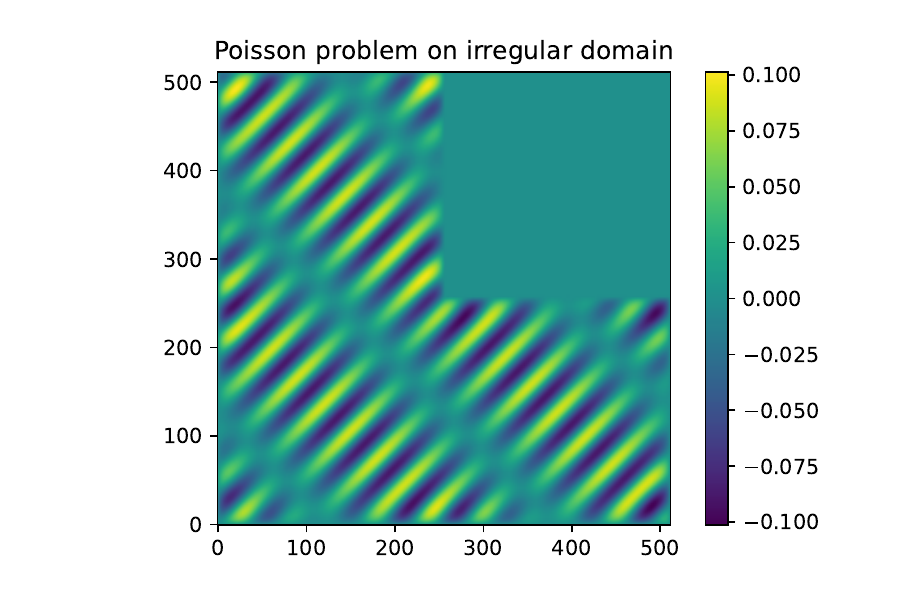} 
    \end{minipage}
    \hfill
    \begin{minipage}{0.45\textwidth}
        \includegraphics[width=\textwidth]{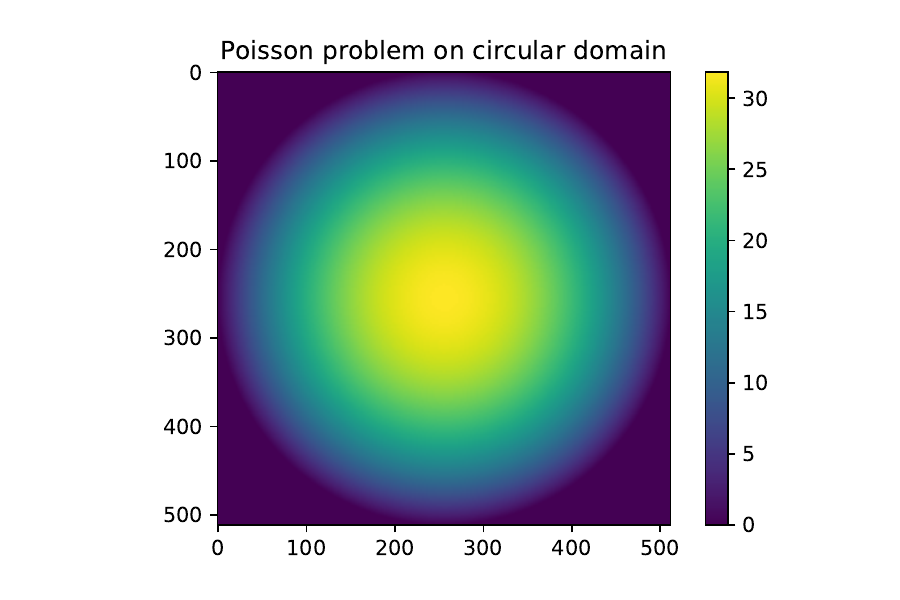} 
    \end{minipage}
    \caption{Here we demonstrate solving the obtained linear system using the projection method on an L-shaped and circular domain, two examples where the problem cannot be expressed as the direct sum over two one-dimensional domains. The linear system is constructed by first constructing the two-dimensional Laplacian matrix with periodic boundary conditions, and then multiplying on the left by a diagonal projection matrix containing 1 in entries corresponding to interior grid points. We then add to this matrix a projector onto the boundary to enforce a homogeneous Dirichlet boundary condition. We then add a projector onto the exterior portion, so that obtained linear system is not singular, placing zeros in the corresponding locations of the right hand side vector. For the L-shaped domain we a product of $\sin$ and $\cos$, and use a constant for the circular domain to generate the rhs vector. The code generating these plots can be found on Github \cite{KharaziCode}. }    
    \label{fig:proj-method-for-nonTP-domain}

\end{figure}

Now, we show how more general boundary conditions can be realized using this oracle and a phase kickback. Consider some arbitrary, suitably normalized initial starting state,
\begin{equation}
    \ket{x} = \sum_{i \in [N^d]}x_i \ket{i},
\end{equation}
and append an additional register of 2 ancilla qubits $\ket{b} \rightarrow \ket{b}\ket{00}$. Then, applying the oracle $\textsc{dom}$ to $\ket{b}\ket{00}$ we find
\begin{equation}
    \textsc{dom}\ket{x}\ket{00} = \ket{x_{int}}\ket{00} + \ket{x_{\partial}}\ket{01} + \ket{x_{ext}}\ket{10}.
\end{equation}
To obtain a block encoding of the operator on the interior, we will use a phase kickback oracle to cancel out matrix elements on the boundary and exterior. Since the boundary and exterior nodes are marked with a set bit, we can implement the desired sign flip by checking the parity of the input string. We will assume access to the function $f(b_0,b_1) = b_0 \oplus b_1$ which we implement via the unitary
\begin{equation}
    U_f \ket{b_0,b_1} = (-1)^{f(b_0,b_1)}\ket{b_0,b_1},
\end{equation}
where $\ket{\perp} \in \text{span}\{\ket{01},\ket{10}\}$. Now, assuming access to the block encoding of the interior $L$ operator with periodic boundary conditions on $[0,L]^d$, we can implement the operator restricted to the interior points. 

\begin{figure}
    \centering
    \scalebox{0.8}{
    \begin{quantikz}
        \lstick{$\ket{0}$}&\gate{H}
        \gategroup[7,steps=9,style={dashed,rounded
        corners,fill=blue!20, inner
        xsep=2pt},background,label style={label
        position=below,anchor=north,yshift=-0.2cm}]{{$(2(\lambda + \beta), m+l)$-BE$({\mathcal{L}})$}}
        &&&&&\octrl{3}&&\ctrl{3}&\gate{H}&\bra{0} \\
        \lstick{$\ket{0}$}&\gate{H}&\octrl{1}&&\octrl{1}& &\ctrl{1}&&\ctrl{1}&\gate{H}&\bra{0}\\
        \lstick{$\ket{0}$}&\gate{H}&\octrl{1}&&\ctrl{2}& &&&&\gate{H}&\bra{0}\\
        \lstick{$\ket{00}$}&\gate[2]{\textsc{dom}}&\vqw{1}&&\gate{U_f}&&\vqw{1}&&\gate{U_{f'}}\vqw{1}&\gate[2]{\textsc{dom}^\dagger}&\bra{00}\\
        \lstick{$\ket{x}$}&&\gate[2]{U_L}&&\gate[2]{U_L}&&\gate[3]{U_B}&&\gate[3]{U_B}&&(\frac{1}{2\lambda}\Pi_{\Omega}\mathcal{L} + \frac{1}{2\beta}\Pi_{\partial\Omega}\mathcal{B})\ket{x} \\
        \lstick{$\ket{0}$}&& &\bra{0}\ket{0}&&\bra{0}\\
        \lstick{$\ket{0}$}&&&&&&&\bra{0}\ket{0}&&\bra{0}
    \end{quantikz}
    }
    \caption{Quantum circuit to implement $(2(\lambda + \beta), m+l)-$block encoding of a finite difference approximation to the differential operator with arbitrary boundary conditions. $\textsc{dom}$ is an oracle which marks each point $\ket{x}$ in the $[N^d]$ with a two-bit label $\ket{x}\ket{b_0,b_1}$, with $'00'$ indicating interior points, $'01'$ {indicating} boundary points and $'10'$ {indicating} exterior points. Here, $U_L$ is a block encoding of the differential operator defined with periodic boundary conditions on $[0,L]^d$ in matrix form and $U_B$ is a block encoding of the Dirichlet, Neumann, or Robin boundary condition in matrix form.  $U_f$ implements the unitary transformation $\ket{b_0,b_1} \rightarrow (-1)^{b_0\oplus b_1}\ket{b_0,b_1}$ and $U_{f'}$ implements $\ket{b_0,b_1}\rightarrow (-1)^{b_0\oplus 1 \oplus b_1}\ket{b_0,b_1}$, to cause the summed operators to be projected onto the appropriate boundary modes. Because $\textsc{dom}$ is not restricted to marking points that correspond to a domain with orthogonal boundary faces, this procedure can be used to implement a broad family of boundary conditions.}
    \label{fig:enter-label}
\end{figure}

Now, consider the controlled operation
\begin{equation}
    \textsc{sel} = I\otimes U_{L}\otimes \ket{00}\bra{00} + U_{f}\otimes U_{L}\otimes \ket{01}\bra{01} + I\otimes U_{B}\otimes \ket{10}\bra{10} + U_{f'}\otimes U_{B}\otimes \ket{11}\bra{11}
\end{equation}
where $U_{L}$ is the block encoding of the operator on the interior and $U_{B}$ is the block encoding of the operator given on the boundary.
{Then}, we have that 
\begin{equation}
    \textsc{sel}\cdot\textsc{dom}\ket{b}\ket{00} = \sum_{i \in \Omega} \frac{b_i}{\lambda}{L}\ket{i}\ket{0}(\ket{0} + \ket{\perp}) + \sum_{i \in \partial{\Omega}}\frac{b_i}{\beta} {B}\ket{i}\ket{1}(\ket{0} + \ket{\perp}) + \sum_{i \notin \Omega \cup \partial \Omega}b_i \ket{i}\ket{2},
\end{equation}
which we write in the simplified form
\begin{equation}
    (H\otimes H \otimes \textsc{dom}^\dagger)\cdot\textsc{sel}\cdot(H\otimes H\otimes \textsc{dom})\ket{b}\ket{00} =  \sum_{i \in \Omega} b_i{L}\ket{i}\ket{0}\ket{0} + \sum_{i \in \partial{\Omega}} b_i{B}\ket{i}\ket{1}\ket{0} + \sum_{i \notin \Omega \cup \partial \Omega} b_i\ket{i}\ket{2} +\ket{\perp},
\end{equation}
where the $\ket{0}$ is the flag to indicate the successful application of the block encoding unitary. 

Writing the ancilla states in binary {notation,} we notice that the states indicating the boundary and exterior nodes correspond to $\ket{01}$ and {to} $\ket{10}$ respectively. We will assume access to the function $f(b_0,b_1) = b_0 \oplus b_1$, which we implement via the phase kickback oracle
\begin{equation}
    U_f \ket{b_0,b_1} = (-1)^{f(b_0,b_1)}\ket{b_0,b_1}.
\end{equation}
Then, to obtain a block encoding of the interior portion of the operator, we can {undertake} a linear combination of block encodings $I\otimes U_{L}$ and $U_{f}\otimes U_{L}$, with $\textsc{prep} = \frac{1}{\sqrt{2}}(\ket{0}+\ket{1})$ and $\textsc{sel} = (I\otimes U_L)\otimes\ket{0}\bra{0} + (U_f\otimes U_L) \otimes \ket{1}\bra{1}$. 
For an arbitrary state $\ket{x}$, we may {then} write
\begin{align*}
    \ket{x}\ket{00}\ket{0} &\overset{I\otimes I \otimes H}{\rightarrow} \frac{1}{\sqrt{2}}\ket{x}\ket{00}(\ket{0} + \ket{1})\\
    &\overset{\textsc{dom}}{\rightarrow} \frac{1}{\sqrt{2}}\frac{1}{N_x}\ket{x_{\Omega}}\ket{00} + \ket{x_\partial}\ket{01} + \ket{x_{out}}\ket{10})(\ket{0} + \ket{1})\\
    &\overset{\textsc{sel}}{\rightarrow} \frac{1}{\sqrt{2}}\frac{1}{N_x}U_{L}(\ket{x_{\Omega}}\ket{00} + \ket{x_\partial}\ket{01} + \ket{x_{out}}\ket{10})\ket{0} +U_L(\ket{x_{\Omega}}\ket{00} - \ket{x_\partial}\ket{01} - \ket{x_{out}}\ket{10})\ket{1}\\
    &\overset{I\otimes I \otimes H}{\rightarrow} \frac{1}{2}U_{L}\ket{x_{\Omega}}\ket{00}\ket{0} +\ket{\perp}.
\end{align*}
Therefore, these operations implement a $(2\lambda, m+3)$ block encoding of the interior operator. We note that this operator only implements the differential operator on the interior nodes and projects everything else to zero.  We now need to add the boundary operators.

This can be done similarly, replacing $U_f$ with $U_{f'}$ (which can be implemented by conjugating the control with $X$-gates on the least significant bit) which performs $U_f\ket{b_0,b_1} = (-1)^{f(b_0\oplus 1, b_1)}\ket{b_0,b_1}$, where the {effect of the} $\oplus1$ on $b_0$ bit is to place a minus sign on the nodes on {both} the interior and the exterior, while leaving {the nodes} on the boundary unaffected. If we consider a generic boundary condition and block encode the corresponding differential operator $U_B$, then we can implement the operator for just the boundary via the linear combination of block encodings $\frac{1}{2}\left( I\otimes U_{B} + U_{f'}\otimes U_{B}\right)$, analogously to the interior operator. Therefore, the interior and boundary operator can {now} be implemented via the linear combination of block encodings {according to}
\begin{equation}
    U_\mathcal{L} \propto \frac{1}{2\lambda}( U_L\otimes I + U_L \otimes U_f ) + \frac{1}{2\beta}(U_B\otimes I + U_{B}\otimes U_{f'}).
    \label{eq:lcu dom oracle}
\end{equation}
This formula is analogous to and generalizes those obtained in the above sections, but the implementation of the reflection operation in this case required an additional register and oracle to label each grid point. This procedure allows one to begin with the efficient block encodings of the periodic operators {that} we found in the above section, and {to} project them onto the corresponding subsets of $[0,L]^d$ so that in total the obtained operator acts as the desired differential equation and boundary conditions on the ``physical" domain $\Omega \cup \partial \Omega$, while acting trivially on $\Omega_L^d \backslash\{\Omega \cup \partial \Omega\}$. Finally, we note that if the initial state $\ket{b}$ has no support on $[0,L]^d\backslash\Omega$, which would correspond to a physically relevant starting state, then the success probability of implementing this block encoding can be as large as $O(1/d)$.

\section{\label{sec:GenElliptic} Block encoding elliptic operator with multi-particle potential}
To block encode a general elliptic operator {for a many-body system}, we need to be able to evaluate gradients of potential functions. {In particular, we wish} to construct a block encoding of the operator $\nabla V \cdot \nabla$. This term is sometimes referred to as the ``convective term" and describes the transport of quantities such as mass under the influence of a force $\nabla V$. {Such operators} arise in dynamical simulations {described by the Fokker-Planck and Smoluchowski equations, and steady-state formulations of these, i.e., the backward and forward Kolmogorov equations, as well as in the closely related Black-Scholes equation.}  We focus {here} on the particular case where the potential function $V$ is given as the sum over two-particle Lennard-Jones potentials
\begin{equation}
    \label{eq:Lennard Jones}
    V_{LJ}(r_{ij}) = 4\varepsilon\left(\left(\frac{\sigma}{r_{ij}}\right)^{12} - \left(\frac{\sigma}{r_{ij}}\right)^6\right).
\end{equation}
{Here} $\varepsilon$ is the depth of the potential, $\sigma$ is the particle size, and $r_{ij} = \norm{\mathbf{x}^i - \mathbf{x}^j}$ is the 2-norm distance between particles $i$ and $j$ in $[\eta]$. We choose this potential {because} it incorporates the features of a wide range of potentials, in that it requires the evaluation of inverse powers of the interparticle distance $r$ and has both attractive and repulsive regimes. The Lennard-Jones potential is a very common effective potential used in molecular dynamics simulations, and is therefore a natural setting {for the} study {of quantum algorithms for} many-body simulations. However, the constructions we present in this section could be applied to many other multi-particle potential functions in a natural way.

This section will proceed as follows. First we will briefly review a block encoding of the discrete position operator, i.e. the operator which returns the location of a grid point
\begin{equation}
    X\ket{i} = i\ket{i}.
\end{equation}
Then, we show how to compute the difference of position of two particles using, i.e.
\begin{equation}
    R\ket{i}\ket{j} = (i-j)\ket{i}\ket{j}.
\end{equation}
Then, we show how to compute the squared distance between two particles in $d$-dimensions $\ket{\bs{\alpha}}$, $\ket{\bs{\beta}}$, i.e. 
\begin{equation}
    R^2\ket{\bs{\alpha}}\ket{\bs{\beta}} = \sum_{i\in[d]}\left(\alpha_i - \beta_i\right)^2 \ket{\bs{\alpha}}\ket{\bs{\beta}}
\end{equation}
which is the squared Euclidean distance between the two particles. Once this diagonal operator is block encoded, we can then perform a polynomial transformation of these diagonal matrices using QSP to implement a polynomial approximation to multiplication operator $V$. Then, we show a technique to implement a block encoding of the term $\nabla V \cdot \nabla$ combining the techniques of this section with the previous section.

\subsection{Block encoding position operators}
\label{sub:BE pos ops}
Our first goal is to block encode the operator 
\begin{equation}
    X \equiv \sum_{i=0}^{N-1} x(i) \ket{i}\bra{i},
    \label{eq:1d pos}
\end{equation}
where $x(i)$ maps the point $i$ on the grid to the corresponding point $x(i)$ in real space. In the simplest case on $[0,1]$, we take $x(i) = i/(N-1)$. With a periodic boundary, we take $x(i) = \text{min}(1-i/(N-1),i/(N-1))$. Block encoding the corresponding operator is characterized by the following lemma, which can be found in \cite{mukhopadhyayQuantumSimulationFirstQuantized2023}, Appendix F Lemma 24.
\begin{lem}[LCU of diagonal position matrix \cite{mukhopadhyayQuantumSimulationFirstQuantized2023} Lemma 24]
\label{lem:LCU pos}
Let the discrete ``position operator"\\ $X = \operatorname{diag}(0, 1, \ldots, N-1)$ be an $N\times N$ matrix with $N = 2^n$. Then 
\begin{equation}
    X = \frac{N-1}{2}I - \frac{1}{2}\sum_{i=0}^{n-1}2^i Z_{(i)}
\end{equation}
where 
\begin{equation}
    Z_{(i)} = I^{\otimes (n - i - 1)}\otimes Z \otimes I^{\otimes i},
\end{equation}
gives an $(\lambda = N-1,m = \ceil{\log(n+1)})$-BE($X$), the discrete position operator using the LCU oracles
\begin{equation}
    \begin{aligned}
        \textsc{prep}_X\ket{0} &=\frac{1}{\sqrt{2}}\ket{0} +  \frac{1}{\sqrt{2(2^{n}-1)}}\sum_{i=0}^{n-1}\sqrt{2^{i}}\ket{i+1}\\
        \textsc{sel}_X &= \sum_{i=1}^{n}Z_{i}\otimes \ket{i}\bra{i}.
    \end{aligned}
\end{equation}
\label{lem: Block-encoding position operator}
\end{lem}

With access to the block encoding of the position operator, it is straightforward to obtain a block encoding of the difference of positions operator for one dimension, using the circuit in Fig. \ref{fig:difference-of-positions-operator}. This is a subroutine we will use many times when evaluating potentials that are defined as the symmetric distance between two particles.
\begin{figure}[h]
    \centering
    \begin{quantikz}
        \lstick{$\ket{0}$} 
        & \gate{X} 
        \gategroup[4,steps=6,style={dashed,rounded
        corners,fill=blue!20, inner
        xsep=2pt},background,label style={label
        position=below,anchor=north,yshift=-0.2cm}]{$\left(2(N-1), \ceil{\log(n+1)}+1\right)$-BE$(X \otimes I - I \otimes X)$}
        & \gate{H} & \octrl{1} & \ctrl{1} & \gate{H} & \gate{X} & \rstick{$\bra{0}$}\\
        \lstick{$\ket{0}^{\otimes \ceil{\log(n+1)}}$}&\gate[1]{\textsc{prep}_X} & &\oslash\vqw{1} &\oslash\vqw{2}& &\gate[1]{\textsc{prep}^\dagger_X} & \rstick{$\bra{0}^{\otimes \ceil{\log(n+1)}}$} \\
        \lstick{$\ket{i}_n$} & & &\gate{Z_l} & & & & \\
        \lstick{$\ket{j}_n$} & & & &\gate{Z_l} & &\hphantomgate{} &
    \end{quantikz}
    \caption{Circuit for the block encoding of the difference of positions operator utilizing$\ceil{\log(n+1)}+1$ ancilla qubits and the LCU subroutines from the single particle positions operators. Controlled on the state of the $\ceil{\log(n+1)}+1$ ancilla register, we apply a single qubit Pauli $Z$ operator to the $l$th qubit of the $n$ qubit registers encoding the positions in the grid. This block encoding applies the operation $\ket{i}\ket{j} \rightarrow \frac{(i-j)}{2(N-1)}\ket{i}\ket{j}$ whenever all of the ancilla qubits are in the $\ket{0}$ state.}
    \label{fig:difference-of-positions-operator}
\end{figure}

We can construct block encodings of the squared distance $\left(R_{k}^{i,j}\right)^2$, between particles $i$ and $j$ in dimension $k$, using the block encodings of the difference of position operators. Since by definition the operator $X_i^j$ is just $X$ on every register except the {register} encoding the {$i$th's particle's $j$th coordinate, and acts} trivially elsewhere,  {it is mathematically equivalent to} a block encoding of $X$ applied only to that register. 
It is therefore sufficient to obtain a block encoding of $X_i^j$, given access to a block encoding of $X$. Using an additional ancilla qubit to form the LCU
\begin{equation}
    U_{X_k^i-X_k^j} = \frac{1}{2}\left(U_{X_k^{i}}-U_{X_k^{j}}\right),
\end{equation}
which is implemented as an $(2(N-1), \ceil{\log(n+1)}+1)$ block encoding, see e.g. Fig. \ref{fig:difference-of-positions-operator}.

An immediate corollary of this is the block encoding of the squared distance between two particles in $d$ dimensions. We denote the matrix $\mathbf{R}^2_{i,j}$ to mean the finite-dimensional operator which acts on the $2nd$ qubit registers containing particles $i$ and $j$ which performs the following operation
\[
    \mathbf{R}^2_{i,j}\ket{\bs{\alpha}^i}\ket{\bs{\alpha}^j} = \sum_{k \in [d]}\left(\alpha^i_k - \alpha^j_k\right)^2\ket{\bs{\alpha}^i}\ket{\bs{\alpha}^j},
\]
for every basis state of the $2nd$ qubits. This matrix can be block encoded by a linear combination of a product of block encodings 
\begin{equation}
    \frac{1}{4}\sum_{k=0}^{d-1}\left(U_{X^i_{k}}-U_{X^j_{k}}\right)^2
\end{equation}
where $U_{X^i_{k}}$ is the block encoding of the position operator that acts non-trivially on the $i$th particle's $k$th register. This corresponds to first performing the linear combination of block encodings $U_{X^i_{k}}-U_{X^j_{k}}$, then forming the product of block encodings, then summing this up for each dimension. The cost to form this is characterized by the following theorem.

\begin{figure}[h]
    \centering
    \begin{quantikz}
    \lstick{$\ket{0}^{\otimes2}$}&\gate{S^2}
    \gategroup[3,steps=4,style={dashed,rounded
        corners,fill=blue!20, inner
        xsep=2pt},background,label style={label
        position=below,anchor=north,yshift=-0.2cm}]{{$(4(N-1)^2,\ceil{\log(n+1)} + 3)$-BE$\left(\left(R_k^{i,j}\right)^2\right)$}}
        &\gate{S^{-1}}\vqw{1}& &\gate{S^{-1}}\vqw{1}&\rstick{$\bra{0}^{\otimes 2}$}\\
    \lstick{$\ket{0}^{\otimes \left(\ceil{\log(n+1)}+1\right)}$}&\gate[2]{U_{X_k^i-X_k^j}}&\octrl{0}&\gate[2]{U_{X_k^i-X_k^j}}&\octrl{0}&\rstick{$\bra{0}^{\otimes(\ceil{\log(n+1)}+1)}$}\\
    \lstick{$\ket{b}$}& & & &&\rstick{$\frac{\left(R_{k}^{i,j}\right)^2}{4(N-1)^2}\ket{b}$}
    \end{quantikz}
    \caption{Quantum circuit for the product of two block encodings applied to the difference-of-positions operator. Post-selecting on the ancilla qubits being measured in the all zero state, the squared difference of position operator for particles $i$ and $j$ in direction $k \in [d]$ is applied to the input state. Using an LCU over $k\in[d]$ of block encodings of $(R_k^{i,j})^2$, we obtain the squared Euclidean distance in $d$ dimensions.}
    \label{fig:r^2 op}
\end{figure}

\begin{cor}[Cost to block encode $\mathbf{R}_{i,j}^2$ operator]
\label{cor:BE R^2}
There exists a block encoding of the squared Euclidean distance between every pair of points on two grids of $N^d$ nodes, with subnormalization $\lambda = 4d(N-1)^2$ using $m = \ceil{\log(n+1)}+ \ceil{\log(d)}+3$ ancilla qubits and this block encoding can be implemented with $O(dn)$ Toffoli or simpler gates. 
\end{cor}
\begin{proof}
    The cost to block encode this operator can be obtained by combining the costs to construct the circuits in figures \ref{fig:difference-of-positions-operator},\ref{fig:r^2 op}, and \ref{fig:d-dim r^2}. The circuit in Fig. \ref{fig:difference-of-positions-operator} requires $O(n)$ Toffoli gates to implement the individual position operators, which applied in a controlled manner approximately doubles the number of gates, for a total of $O(4n)$ Toffoli gates. The circuit in Fig. \ref{fig:r^2 op} applies the circuit in Fig. \ref{fig:difference-of-positions-operator} twice, for a total of $O(8n)$ Toffoli gates. These are then applied in a controlled manner $d$ times to construct the circuit in Fig. \ref{fig:d-dim r^2} which would therefore require $O(16dn)$ Toffoli gates to construct the $d$-dimensional squared Euclidean distance operator. Since we can reuse the ancilla qubits involved in the block encodings of the individual terms in the positions operators, the circuit in Fig. \ref{fig:d-dim r^2} only requires an additional $\ceil{\log(d)}$ ancilla qubits over that of Fig. \ref{fig:r^2 op}, for a total of $\ceil{\log(n+1)} + \ceil{\log(d)} + 3$ ancilla qubits. 
\end{proof}

\begin{figure}[h]
    \centering
    \begin{quantikz}
    \lstick{$\ket{0}^{\otimes \ceil{\log(d)}}$}&\gate{D(d)}
    \gategroup[3,steps=4,style={dashed,rounded
        corners,fill=blue!20, inner
        xsep=2pt},background,label style={label
        position=below,anchor=north,yshift=-0.2cm}]{{$(4(N-1)^2,\ceil{\log(n+1)} + \ceil{\log(d)}+3)$-BE$\left(\left(\mathbf{R}_{ij}^2\right)\right)$}}
        &\oslash\vqw{1} & &\gate{D^\dagger(d)}&\rstick{$\bra{0}^{\otimes\ceil{\log(d)}}$}\\
    \lstick{$\ket{0}^{\otimes \left(\ceil{\log(n+1)}+3\right)}$}&&\gate[2]{\left(R_k^{i,j}\right)^2}&&&\rstick{$\bra{0}^{\otimes(\ceil{\log(n+1)}+3)}$}\\
    \lstick{$\ket{b}$}&\qwbundle{sys} & & &&\rstick{$\frac{\mathbf{R}_{ij}^2}{4d(N-1)^2}\ket{b}$}
    \end{quantikz}
    \caption{Quantum circuit for the $d$-dimensional distance of positions operator between two particles $i$ and $j$. The top ancilla register prepares the uniform superposition over $d$ states using the Diffusion operator $D(d)$. Controlled on the state of this register, we apply the one-dimensional squared difference of position operator on the system register. Then, upon measuring the all-zero state for all of the ancilla qubits, we will have applied the $d$-dimensional Euclidean distance operator between particles $i$ and $j$ to the system register. Overall, this results in an $\alpha = 4d(N-1)^2$, $m = \ceil{\log(n+1)} + \ceil{\log(d)} + 3$ block encoding of $R$ in the   product of two block encodings applied to the difference-of-positions operator. Post-selecting on the ancilla qubits being measured in the all zero state, the squared difference of position operator for particles $i$ and $j$ in direction $k \in [d]$ is applied to the input state. }
    \label{fig:d-dim r^2}
\end{figure}

\subsection{Block encoding potential function}
\label{subsec:BE potential fn}
The operator we block encoded in the previous section is a diagonal matrix which encodes the squared Euclidean distance between two particles occupying positions in the $d$-dimensional grid. Combining the block encoding of corollary \ref{cor:BE R^2} with a polynomial approximation to the Lennard-Jones $3$-$6$ potential, we can use QSP to operate on this diagonal matrix to construct a polynomial approximation to the Lennard-Jones $6$-$12$ potential. Generally, our goal is to implement a polynomial transformation to the block encoding of the distance matrix to approximate the Lennard-Jones potential,
\begin{equation}
    \mathbf{R}^2_{ij} \overset{\text{poly}}{\longrightarrow} \widetilde{V}(\mathbf{R}^2_{ij}),
\end{equation}
where poly is some polynomial approximation to $\frac{1}{r^6}-\frac{1}{r^3}$ we call $\widetilde{V}$. $\widetilde{V}$ is a multiplication operator and is diagonal in the standard basis, 
\begin{equation}
    \widetilde{V}(\mathbf{R}^2_{ij}) = \sum_{\bs{\alpha},\bs{\beta} \in [N^d]} \widetilde{V}\left(\norm{\bs{\alpha} - \bs{\beta}}^2\right)\ket{\bs{\alpha}}_i\ket{\bs{\beta}}_j\bra{\bs{\alpha}}_i\bra{\bs{\beta}}_j.
\end{equation}

\begin{figure}[h]
    \centering
    \includegraphics[width=0.6\textwidth]{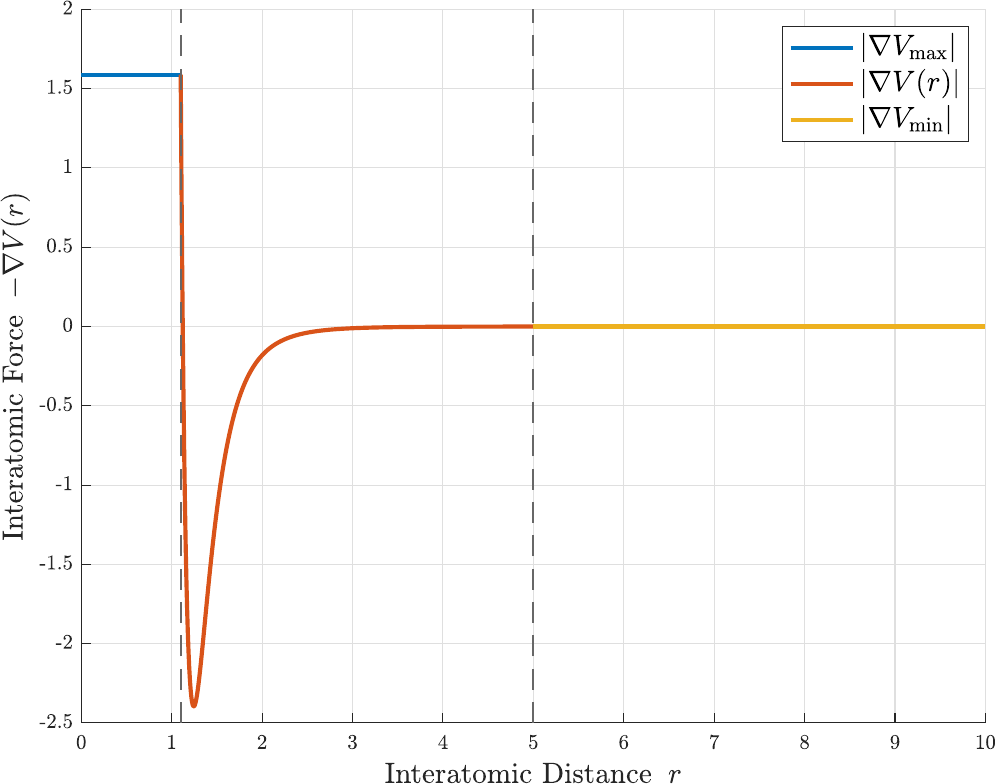}
    \caption{Depiction of the gradient of the piecewise Lennard-Jones potential in $1d$. The dashed lines demarcate the interval where we compute the polynomial approximation to the potential $V(r)$, outside of which we continuously extend as a constant.}
    \label{fig:piecewise-LJ}
\end{figure}

Using the QSP subroutine on the block encodings for the squared distance, we can obtain an approximation to the potential operator as a function of the squared distance. For inverse power law potentials, the following lemma characterizes the convergence in terms of polynomial degree and a cut-off parameter $\delta>0$ used to make the inverse power function well-defined near zero.
\begin{lem}[Convergence rate for polynomial approximations to inverse power law potentials (Ref. \cite{gilyenQuantumSingularValue2019} Corollary 67)]
    Let $\delta, \epsilon \in (0,1/2], c>0$ and $f(x) = \frac{\delta^c}{2}x^{-c}$, then there exist even/odd polynomials $P, P' \in \mathbb{R}[x]$ such that $||P-f||_{[\delta, 1]} \leq \epsilon$, $||P||_{[-1,1]} \leq 1$, and similarly $||P'-f||_{[\delta, 1]}\leq \epsilon$ and $||P'||_{[-1,1]} \leq 1$, moreover the degree of the polynomials are $O\left(\frac{\text{max}[1,c]}{\delta}\log(1/\epsilon)\right)$.
    \label{lem: convRate for invPow}
\end{lem}

A common choice of the cutoff is $\delta = O(1/N)$, so that the polynomial approximation rescales the Lennard-Jones potential by a factor of $O(N^{-12})$. This factor of $O(N^{-12})$ will need to be accounted for, as it does not correspond to a global scaling factor. Lemma \ref{lem: convRate for invPow} also implies that the polynomial degree to block encode the rescaled Lennard-Jones potential $\frac{\delta^{12}}{2} V_{LJ}$, is $ O\left(12 N \log(1/\epsilon)\right)$. In some scenarios it may be more natural to choose a cutoff parameter $\delta \sim \sigma/C$ where $C>1$ is a constant. Since the polynomial degree is directly related to the number of queries to the block encoding of $\mathbf{R}^2$, the cost to block encode a single term of the potential is $\widetilde{O}(N)$ queries to the block encoding of $\mathbf{R}^2$, the block encoding of $\mathbf{R}^2$ required $O(d)$ queries to the difference of position operator, which can each be implemented with $O(\log(N))$ Toffoli or simpler gates. Therefore, the total cost to implement a single term in the potential function $V_{LJ}(\mathbf{R}^2_{ij})$ to $\epsilon$ accuracy is 
\begin{equation}
    \widetilde{O}(dN\log(N)\log(1/\epsilon))
    \label{eq:BE V}
\end{equation}
Toffoli or simpler quantum gates.

Clearly the Lennard-Jones rapidly diverges as $r\downarrow 0$. Because this divergence occurs with such a large polynomial degree, it is difficult to obtain uniform and numerically stable polynomial approximations, even with a cutoff. In practice, very high degree polynomials are required. It is worthwhile to consider other avenues beyond finding a global polynomial approximation, such as using different approximating polynomials for different subsets of the function's domain. We present an algorithm that allows one to selectively apply various transformations to subspaces of the block encoded matrix, and present an algorithm to mark the relevant subspaces using the method of inequality testing.

\begin{figure}[h]
    \centering
    \scalebox{0.55}{
    \begin{quantikz}
        \lstick{$\ket{i}_n$} & \gate[3]{\comp}
        \gategroup[7,steps=18,style={dashed,rounded
        corners,fill=blue!20, inner
        xsep=2pt},background,label style={label
        position=below,anchor=north,yshift=-0.2cm}]{{$\dom$}}
        & \gate[2]{\mathrm{SWAP}} & \gate[2]{\mathrm{DIFF}} & \swap{3} & \gate[3]{\comp} & & \gate[3]{\comp^\dag} & \swap{3} & & & \swap{3} & \gate[3]{\comp} & & \gate[3]{\comp^\dag} & \swap{3} & \gate[2]{\mathrm{DIFF}^\dag} & \gate[2]{\mathrm{SWAP}} & \gate[3]{\comp^\dag} & \rstick{$\ket{i}_n$}\\
        \lstick{$\ket{j}_n$} & & & & & & & & & & & & & & & & & & & \rstick{$\ket{j}_n$} \\
        \lstick{$\ket{0}$} & & \octrl{-1} & \swap{2} & & & \octrl{4} & & & & & & & \octrl{3} & & & \swap{2} & \octrl{-1} & & \rstick{$\ket{0}$} \\
        \lstick{$\ket{0}_n$} & \gate{S^{(r_\text{min})}} & & & \targX{} & & & & \targX{} & \gate{S^{(r_\text{min})\dag}} & \gate{S^{(r_\text{max})}} & \targX{} & & & & \targX{} & & & \gate{S^{(r_\text{max})\dag}} & \rstick{$\ket{0}_n$}\\
        \lstick{$\ket{0}$} & & & \targX{} & & & & &\hphantomgate{} & & & & & & & & \targX{} & & & \rstick{$\ket{0}$}\\
        \lstick{$\ket{0}$} & & & & & &\hphantomgate{} & &\hphantomgate{} & & & & & \targ{} & \ctrl{1} & & & & & \rstick[2]{domain\\flags} \\
        \lstick{$\ket{0}$} & & & & & & \targ{} & & & & & & & & \targ{} & & & & &
    \end{quantikz}
    }
    \caption{An explicit circuit for the domain {oracle} circuit $\dom$ utilizing $n+2$ ancillary qubits. The two flag qubits return $\ket{00}$ if $|i - j| \in [0,r_\text{min}]$, $\ket{01}$ if $|i - j| \in (r_\text{min}, r_\text{max}]$, and $\ket{10}$ if $|i - j| \in (r_\text{max}, 1]$.}
    \label{fig:dom-circuit}
\end{figure}

Using a comparator circuit, we build a circuit that we denote as $\dom$ for assessing which part of the piecewise domain our input is in. First, we first partition our domain as 
\begin{equation}
    \Omega = [0,1] = [0,r_\text{min}) \sqcup [r_\text{min}, r_\text{max}] \sqcup (r_\text{max}, 1] \,.
\end{equation}
Within the interval $[r_\text{min},r_\text{max}]$, we compute a polynomial approximation to $V(r)$ that converges to $V$ at an exponential rate in the polynomial degree. Outside this interval, we continuously extend $V(r)$ to a constant function (see Figure~\ref{fig:piecewise-LJ}). The use of comparator circuits in quantum computing was introduced in Ref. \cite{sandersBlackBoxQuantumState2019}. In {that} work a quantum circuit {was constructed which} implements the following inequality relations:
\begin{align}
    \comp \ket{i}\ket{j}\ket{0} &:= \begin{cases}
        \ket{i}\ket{j}\ket{0} & \text{if } i < j \,,\\
        \ket{i}\ket{j}\ket{1} & \text{if } i \geq j \,,
    \end{cases}
\end{align}
where $\ket{i}$ and $\ket{j}$ are $n$-qubit computational basis states. This is commonly refered to as ``inequality testing" in the literature. The comparator circuit is obtained by modifying a quantum addition circuit to compute the most significant bit of the {bit-wise} subtraction $i - j$. In both one's and two's complement (i.e. unsigned and signed) binary arithmetic, bit subtraction can be rewritten in terms of bit addition using the identity \cite{cuccaro2004new}
\begin{equation}
    j - i = (j' + i)' \,,
\end{equation}
where $(\cdot)'$ is bitwise complementation. Using the quantum addition circuit given in \cite{gidneyHalvingCostQuantum2018}, the $\comp$ routine can be implemented using only $n$ Toffoli gates.

We implement the piecewise polynomial using a domain oracle defined as 
\begin{align} \label{eq:dom_oracle}
    \dom \ket{i}_n\ket{j}_n\ket{0}_2 := \begin{cases}
        \ket{i}_n\ket{j}_n\ket{0}_2 & \text{if } 0 \le |i - j| \le r_\text{min} \,,\\
        \ket{i}_n\ket{j}_n\ket{1}_2 & \text{if } r_\text{min} < |i - j| \le r_\text{max} \,,\\
        \ket{i}_n\ket{j}_n\ket{2}_2 & \text{if } r_\text{max} < |i - j| \le 1 \,,
    \end{cases}
\end{align}
where $r_\text{min}, r_\text{max} \in [N]$ are the closest bit-string representations to, respectively, the minimum and maximum cut-off radii. This is done by combining quantum addition with the $\comp$ oracle. First, let $\textsc{diff}$ be defined as 
\begin{align}
    \textsc{diff}:\ket{i}_n\ket{j}_n &\rightarrow \ket{i}_n \ket{i - j}_n \,.
\end{align}
which can be implemented using a quantum addition circuit. The piecewise operations are then implemented using the circuit in Figure~\ref{fig:piecewise-potential}. The unitaries that implement the potential function gradients are the subject of the following sections.

\begin{figure}[h]
    \centering
    \begin{quantikz}
        \lstick{$\ket{i}_n$} & \gate[4]{\dom}
        \gategroup[4,steps=5,style={dashed,rounded
        corners,fill=blue!20, inner
        xsep=2pt},background,label style={label
        position=below,anchor=north,yshift=-0.2cm}]{{$\widetilde{V}$}}
        & \gate[2]{V_\text{max}} & \gate[2]{ V} & \gate[2]{V_\text{min}} & \gate[4]{\dom^\dag} & \\
        \lstick{$\ket{j}_n$} & & & & & & \\
        \lstick{$\ket{0}$} & & \octrl{-1} & \octrl{-1} & \ctrl{-1} & & \\
        \lstick{$\ket{0}$} & & \octrl{-2} & \ctrl{-2} & \octrl{-2} & &
    \end{quantikz}
    \caption{Circuit implementation of the domain-dependent application of the block-encodings. {$\dom$ represents the domain oracle, (\ref{eq:dom_oracle}), and the gates $V_\text{max}$, $V$, and $V_\text{min}$ 
    denote the block encodings of the upper potential cutoff $V_\text{max}$, the two-particle potential $V$, and the lower potential cutoff $V_\text{min}$, respectively.}}
    \label{fig:piecewise-potential}
\end{figure}

\begin{lem}[Cost to construct $\dom$ oracle]
    The number of Toffoli or simpler gates needed to implement the piecewise domain oracle is $O(dn)$.
\end{lem}
\begin{proof}
From Ref. \cite{sandersBlackBoxQuantumState2019}, we know that $\comp$, when applied to two, $n$ qubit registers requires $n$ Toffoli gates to implement. The controlled swap operation between two $n$ qubit registers will also require $O(n)$ Toffoli, and $\textsc{diff}$ as well. Therefore, all the $n$ qubit operations can be performed with $O(n)$ Toffoli gates. Since we need to perform these comparisons in $d$ dimensions, we therefore require $O(dn)$ Toffoli or simpler gates to implement the $\dom$ oracle.
\end{proof}

Although this introduces some additional gate cost in the construction of the block encoding, this method allows us to greatly improve the polynomial degree in practice, and can improve the dependence of the $\delta \sim N^{-1}$ parameter in the polynomial degree. We note that a similar approach is used to block encode the Coulomb potential in Ref. \cite{Su2021FirstQuant}.

\subsection{Block encoding gradient operator}
\label{subsec:block encode grad V}

Once we have obtained our approximate block encoding of the potential $\widetilde{V}$, we turn to the construction of the gradient.  When the partial derivatives of the potential do not have a clean closed form expression, we will need to evaluate the partial derivatives numerically. To do so we will use a central difference scheme and perform multiple evaluations of the block encoding of $\widetilde{V}$ at shifted coordinates. This approach may seem suboptimal at first, since the block encoding of $V$ is by far the costliest operation we perform. It may be tempting to instead perform matrix multiplication of the block encodings of differential operators obtained in Sec. \ref{sec:Laplacian} applied to the block encoding of $\widetilde{V}$. However, this does not encode the correct operations, as multiplying the diagonal matrix $\widetilde{V}$ on the left with a finite difference matrix does not reproduce the desired behavior. We instead interpret $\nabla V = \mathbf{b}(x)$ as a vector operator, where each element is a diagonal multiplication operator $\partial_i V$ that we form by a linear combination of $\widetilde{V}$ evaluated at shifted values of the distance function weighted by the finite difference coefficients.

We construct the appropriate finite differences by directly shifting the individual position coordinate operators for each dimension and particle register. The appropriate shift can be realized via the similarity transformation
\begin{equation}
    S^{j} X S^{-j} = \frac{1}{N-1}\text{diag}(j, j+1, \ldots,N-1,0, \ldots, j-2, j-1),
\end{equation}
using the techniques from \ref{sub:BE pos ops}, we can straightforwardly form the partially shifted squared difference-of-position operator in $d$-dimensions, 
\begin{equation}
\begin{aligned}
        (\mathbf{R}^{i,j}_{k,l})^2 = (S^l_kX_k^iS_k^{-l} - X^j_k)^2 + \sum_{m \neq k\in[d]} (X_m^i-X_m^j)^2.
        \label{eq:shifted pos ops}
\end{aligned}
\end{equation}
This corresponds to shifting the position operator of particle $i$ in direction $k$ by $l$ when computing the squared distance between particles $i$ and $j$. We do this for every particle $i,j \in[\eta]$ and for every $k \in [d]$ and for every $1\leq |l| \leq a$, so that we have block encodings of the squared Euclidean distance operator, evaluated at the shifted coordinates. The approximate gradient in $d$-dimensions, evaluated using the finite difference coefficients $c_l$, can therefore be expressed
\begin{equation}
    \nabla^i V \approx \sum_{k \in [d]}\sum_{1 \leq |l| \leq a} c_l V\left((\mathbf{R}^{i,j}_{k,l})^2 \right)\otimes \ket{k}.
\end{equation}
We may also write,
\begin{equation}
    V\left((\mathbf{R}^{i,j}_{k,l})^2\right) = \widetilde{V}\left(S^l_k\bs{\alpha}^iS^{-l}_k,\bs{\alpha}^j\right).
\end{equation}
This naturally corresponds to a linear combination of block encodings, and we now detail its implementation.

We append two ancilla registers to encode the dimension and finite difference coefficients of dimension $\ceil{\log(d)}$ and $\ceil{\log(p)}$ respectively. To perform this block encoding we will use QSP over evaluations of the shifted distance matrix in conjunction with an LCU with state preparation oracles,
\begin{equation}
\begin{aligned}
        \textsc{prep}_D\ket{0} &= \sum_{1 \leq |l| \leq a} \sqrt{\frac{c_l}{\beta}}\ket{l}\\
        D(d)\ket{0} &= \frac{1}{\sqrt{d}}\sum_{j \in[d]}\ket{j}
\end{aligned}
\end{equation}
and where $\beta \in O(1)$ is the sum of the absolute values of the first derivative finite difference coefficients. 
To simplify, we introduce the notation
\begin{equation}
    \partial_k^i\widetilde{V}(\bs{\alpha}^i,\bs{\alpha}^j) \equiv \sum_{1\leq |l| \leq a}c_l \widetilde{V}(S^l_k\bs{\alpha}^iS^{-l}_k,\bs{\alpha}^j)
    \label{eq:partial LCU for V}
\end{equation}
to refer to the diagonal matrix that encodes the approximate partial derivative to $V_{ij}$ applied to particle $i$ in direction $k$ using a $2a$-point finite difference stencil with $c_l$ the finite difference coefficients for the first derivative (see Eq. \eqref{eq:FD coeffs d/dx}). The single-particle $d$-dimensional gradient is,
\begin{equation}
    \nabla^i \widetilde{V} = \sum_{k\in[d]} \partial_k^i\widetilde{V}(\bs{\alpha}^i,\bs{\alpha}^j)\otimes \ket{k}.
\end{equation}
We show the quantum circuit implementing the block encoding of this matrix as Fig. \ref{fig:BE-gradV}.
The following lemma characterizes the cost to implement this block encoding. 

\begin{lem}[Cost to block encode $\nabla {V}$ for Lennard-Jones Potential]
    Let $V$ be the two-particle Lennard-Jones potential, let $d$ be the underlying spatial dimension of the system, and $p = 2a-1$ the order of the finite difference stencil we use to evaluate the gradient. Then, there exists a quantum circuit preparing an $(O(dN^{12}), \ceil{\log(n+1)}+2\ceil{\log(d)}+\ceil{\log(p)}+4)$-BE$(\nabla \widetilde{V})$ using $O(d^2p^2 N \log^2(N))$ Toffoli or simpler gates.
    \label{lem:gradV cost}
\end{lem}
\begin{proof}

    We begin by first analyzing the gate complexity to construct Eq. \eqref{eq:partial LCU for V}. Constructing the shifted position operator requires two evaluations of the shift operators per term in the finite difference scheme, which requires $O(n)$ Toffoli or simpler gates each, and one evaluation of the position operator which also requires $O(n)$ Toffoli or simpler gates, so a total of $\sim 3n$ Toffoli gates are needed to block encode the one-dimensional shifted position operator. We then construct the squared Euclidean distance matrix with respect to the shifted coordinate, which will requires $\sim 3dn$ Toffoli or simpler gates. We then evaluate use QSP to approximate $V$, we would like our approximation $\widetilde{V}$ to be accurate to the same order as the finite difference scheme, which using estimate in Eq. \ref{eq:BE V} for the polynomial degree, will require $O(p N n)$ queries to the block encoding of the squared Euclidean distance function. Therefore, to block encode $V$ evaluated at a single shifted coordinate requires $\sim 3d N n^2$ Toffoli or simpler gates. Since there are $2(p-1)$ shifts that need to be evaluated to approximate Eq. \eqref{eq:partial LCU for V} for both particles, we will require an additional $O(6dp^2 N n^2)$ Toffoli or simpler gates. Since the gradient will require repeating this for $2d$ times for both pairs of particles, we obtain a gate cost that is $O(12d^2 p^2N n^2)$ Toffoli or simpler gates. Substituting $n = \log(N)$ and removing the constant factor we obtain reported Toffoli complexity.
    The number of ancilla qubits can be immediately read-out by the circuit in Fig. \ref{fig:BE-gradV}, and the subnormalization factor directly results from the $N^{12}$ factor needed to normalize the Lennard-Jones potential and the factor of $d$ from the LCU over the $d$ terms.
\end{proof}

\begin{figure}[h]
    \centering
    \scalebox{0.8}{
    \begin{quantikz}
        \ket{0}^{\otimes \ceil{\log(d)}}&\qwbundle{${dim}$} \gategroup[5,steps=10,style={dashed,rounded
        corners,fill=blue!20, inner
        xsep=2pt},background,label style={label
        position=below,anchor=north,yshift=-0.2cm}]{{$(\beta \sqrt{d}N^{12}, \ceil{\log(n+1)}+ 2\ceil{\log(d)}+\ceil{\log(p)}+4)$- BE$(\nabla^i V)$}}&\gate{D(d)}&{\oslash}\vqw{1}&&\oslash\vqw{1}&&\cdots&\oslash\vqw{1}&&&\\
        \ket{0}^{\otimes \ceil{\log(p)}}&\qwbundle{${coeff}$}&\gate{\textsc{prep}_D}&\oslash\vqw{2}&&\oslash\vqw{2}&&\cdots &\oslash\vqw{2}&\gate{\textsc{prep}_D^\dagger}&&\bra{0}\\
        \ket{1}&\qwbundle{${phase}$}&\gate{e^{i\phi_0 Z}}&&\gate{e^{i\phi_1 Z}}\vqw{1}&&\gate{e^{i\phi_2 Z}}\vqw{1}&\cdots&&\gate{e^{i\phi_{d} Z}}\vqw{1}&\gate{X}&\bra{0}\\
        \ket{0}&\qwbundle{${anc}$}&&\gate[2]{\left(\mathbf{R}^{i,j}_{k,l}\right)^2}&\octrl{0}&\gate[2]{\left(\mathbf{R}^{i,j}_{k,l}\right)^2}&\octrl{0}&\cdots&\gate[2]{\left(\mathbf{R}^{i,j}_{k,l}\right)^2}&\octrl{0}&&\bra{0}\\
        \ket{b}&\qwbundle{$sys$}&&&&&&\cdots&&&&
    \end{quantikz}
    }
    \caption{ Quantum circuit for block encoding $d$-dimensional gradient for a single term of the potential function evaluated at the shifted coordinates. The phase register is used with QSP to induce the appropriate phase factors $\phi_k$, corresponding to the approximating polynomial $\widetilde{V}$. Upon measuring the $\text{coeff}$, $\text{anc}$, and $\text{phase}$ ancilla registers in the all zero state, we will have prepared a quantum state of the form $\sum_{k\in[d]}\sum_{1\leq|l|\leq a} \frac{1}{N^{12}\beta\sqrt{d}}c_l \widetilde{V}\left(S_{k}^l \bs{\alpha}^iS_{k}^{-l}, \bs{\alpha}^j)\right)\ket{b}\ket{k} \approx \nabla^i \widetilde{V}$, which provides a finite difference approximation to the $d$-dimensional gradient of particle $i$ for the potential $V_{ij}$. Notice that at this stage we do not yet uncompute the $\text{dim}$ register. We verify that this circuit produces the desired output in theorem \ref{thm:gradV output} of appendix \ref{app:BE}.}
    \label{fig:BE-gradV}
\end{figure}

\subsection{Quantum circuit for convective operator}
\label{subsec:convec op BE}

Given the above constructions for block encoding the potential gradient in the two particle case, we now turn to the block encoding of the full convective operator $(\nabla {V})\cdot \nabla$,
\begin{equation}
\label{eq:eta particle convective}
(\nabla V)\cdot \nabla = \sum_{k \in [\eta]}\sum_{l\leq k}\sum_{i\in[d]}\partial_i^k{V}_{lk}\bs{D}_i^k,
\end{equation}
with $V_{lk} = \widetilde{V}(\mathbf{R}^2_{lk})$.
The partial derivatives $\bs{D}_i^k$ on the right are just finite difference approximations to the first derivative and can be formed through the block encodings we construct in Sec. \ref{sec:Laplacian}, but replacing the finite difference coefficients $r_j$ for the Laplacian with the stencil coefficients $c_j$ for the first derivative.
We can now combine the quantum circuit in Fig. \ref{fig:BE-gradV} and the block encoding of the $d$-dimensional gradient operator with the quantum circuit for a multiplication of block encodings as shown in \ref{fig:single convec term}. We first consider the block encoding of a term in Eq. \eqref{eq:eta particle convective}, of the form
\begin{equation}
    \nabla^i V \cdot \nabla^i = \sum_{k \in [d]}\partial_k^i V_{ij} \partial_k^i
\end{equation}

With central difference coefficients for the first derivative operator $c_j$, we define the LCU oracles,
\begin{equation}
\begin{aligned}
    \textsc{prep}_D\ket{0} &= \sum_{j=-a,j\neq 0}^{a}\sqrt{\frac{|c_j|}{\beta}}\ket{j},
\end{aligned}
\end{equation}
where $\beta$ is the subnormalization factor for the first derivative finite difference coefficients. 
We also define
\begin{equation}
\begin{aligned}
    \textsc{sel}^{i,j}_{{k}}(V) &= \text{sign}(c_j)\sum_{1\leq |j| \leq a}U_V(S_{i,k}^{-j}\bs{\alpha}^iS_{i,k}^{j},\bs{\alpha}^j)\otimes \ket{j}\bra{j},\\
    \textsc{sel}^{i}_k(D) &=\sum_{1\leq |j| \leq a} \text{sign}(c_j)S^j_{i,k}\otimes\ket{j}\bra{j}.
\end{aligned}
\end{equation}
\begin{figure}[ht]
    \centering
    \scalebox{0.8}{
    \begin{quantikz}
    \lstick{$\ket{0}^{\otimes2}$}&\gate{S^2}
    \gategroup[5,steps=8,style={dashed,rounded
    corners,fill=blue!20, inner
    xsep=2pt},background,label style={label
    position=below,anchor=north,yshift=-0.2cm}]{{$(\zeta, \ceil{\log(n+1)}+ 2\ceil{\log(d)}+\ceil{\log(p)}+6)$-BE$({\nabla^i V_{ij} \cdot \nabla^i})$}}
    &&&\gate{S^{-1}}\vqw{2}&&&&\gate{S^{-1}}\vqw{2}&\rstick{$\bra{0}^{\otimes 2}$}\\    \lstick{$\ket{0}^{\otimes\ceil{\log(d)}}$}&\gate{D(d)}&\oslash\vqw{1}&&&&\oslash\vqw{1}&\gate{D^\dagger(d)}&&\rstick{$\bra{0}^{\otimes \ceil{\log(d)}}$}\\
    \lstick{$\ket{0}^{\otimes \ceil{\log(p)}}$}&\gate[1]{\textsc{prep}_D}&\gate[2]{\textsc{sel}^{i}_k(D)}&\gate{\textsc{prep}_D^\dag}&\octrl{0}&\gate{\textsc{prep}_D}&\gate[3]{\textsc{sel}_{k}^{ij}(V)}&\gate[1]{\textsc{prep}_D^\dag}&\octrl{0}&\rstick{$\bra{0}^{\otimes \ceil{\log(p)}}$}\\
    \lstick{$\ket{b}$}& & & &\hphantomgate{0} & & &&&\\
    \lstick{$\ket{0}$}&\qwbundle{anc}&&&\hphantomgate{0}&&&&& \rstick{$\bra{0}$}
    \end{quantikz}
    }
    \caption{Quantum circuit for block encoding the finite difference approximation to $\partial_k^i V_{ij} \partial_k^i$ with subnormalization factor $\zeta = \beta^2 d N^{12}  \in O(\frac{1}{N}\norm{\nabla V})$. We note that we have swapped the locations of the system and ancilla registers, denoted by the bottom two wires, as compared to Fig. \ref{fig:BE-gradV}. For an input state $\ket{b}$ in the system register, this circuit prepares $\frac{1}{\zeta}\sum_{k\in [d]}\partial_k^i V_{ij}\mathcal{D}^i_k\ket{b}$ whenever all of the ancilla qubits are measured to be in the 0 state. We remind the reader that $\partial_k^i$ refers to the partial derivative operator for particle $i \in [\eta]$ in spatial direction $k \in [d]$, implemented as a linear combination of applications of block encodings of the potential function $V_{ij}$. This block encoding implements a single ``diagonal" term in the convective operator, summing up over all of the terms $C^{i}= \nabla^i\sum_{j<i}V_{ij}\cdot \nabla^i$ would give the convective operator for particle $i$. Doing this for all pairs, $i,j$ would give the full convective operator.}
    \label{fig:single convec term}
\end{figure}
The above {protocol} implements a quantum circuit which block encodes a single term of the convective operator.

\begin{lem}[Block encoding of a single convective term]
\label{lem:BE d-dim convec}
    The block encoding of a single term of the form with the approximate Lennard-Jones potential $\widetilde{V}$ 
    \begin{equation}
        {C^{ij}}:= \sum_{k\in[d]}\partial_k^i \widetilde{V}\left(\bold{R}^2_{ij}\right) \bs{D}_k^i
    \end{equation}
    using $p-1$ stencil points to approximate the derivative, requires $O(dp)$ queries to the block encoding of $\widetilde{V}$ and $O(dp)$ applications of the shift operators using $\ceil{\log(n+1)}+ 2\ceil{\log(d)}+\ceil{\log(p)}+6$ ancilla qubits and $\widetilde{O}(Nd^2p^2\log^2(N))$ Toffoli or simpler gates, with subnormalization factor $O\left(d N^{12}\right)$.
\end{lem}
\begin{proof}
    The proof is an almost direct application of the quantum circuit in Fig.\ref{fig:single convec term} and lemma \ref{lem:gradV cost} and theorem \ref{thm:PeriodicLaplacian}.
    From lemma \ref{lem:gradV cost}, the gate cost to form the block encoding of the two particle gradient operator of the Lennard-Jones potential in $d$-dimensions was
    $O(d^2 p^2 N \log^2(N))$ Toffoli. The cost to block encode the first derivative operator with $p-1$ stencil points is essentially equivalent to that of the Laplacian matrix. By theorem \ref{thm:PeriodicLaplacian}, we need $O(d p \log(N))$ Toffoli or simpler gates to perform the block encoding. We can employ two additional ancilla qubits to perform their product of block encodings, and reusing the ancilla registers for the LCU over the dimension and finite difference coefficients. In all, this circuit therefore requires $\widetilde{O}(d^2 p^2 N \log^2(N) + dp \log(N))$ Toffoli or simpler gates and $\ceil{\log(n+1)} + 2\ceil{\log(d)} + \ceil{\log(p)} + 6$ ancilla qubits and implements a block encoding with subnormalization $\zeta = \beta^2 d N^{12} \in O(dN^{12})$ as $\beta^2 \in O(1)$.
\end{proof}

Now, if we consider the many particle case, with $V$ given as a radially symmetric function of its arguments, the full potential function is given by
\begin{equation}
    V_{tot} = \sum_{i\in [\eta]} \sum_{j<i}V(\bs{\alpha}^i,\bs{\alpha}^j),
    \label{eq:Vtot}
\end{equation}
so that the diagonal portion of the convective operator takes the form
\begin{equation}
    C_{tot} =\sum_{i\in [\eta]}\sum_{j<i}\sum_{k\in[d]}\partial_{k}^i V(\bs{\alpha}^i,\bs{\alpha}^j)\partial_{k}^i.
    \label{eq:discrete-convective}
\end{equation}
With access to the block encoding of $C^{i,j}$, given as Fig. \ref{fig:single convec term} we can obtain a block encoding of the many body convective operator by using controlled applications of the ${C}^{i,j}$ circuit over different registers to obtain a block encoding of $C_{tot}$.

We demonstrate how to construct this block encoding using an linear combination of block encodings, similar to above. Let 
\begin{equation}
    U^{ij}_k \in (\zeta/d, k, \epsilon)\text{- BE}(\partial_k^i V(\bs{x}^i,\bs{x}^j)\partial_k^i).
\end{equation}
We can implement the linear combination of block encodings by selecting over the block encodings of the individual terms in the summand using the select oracle 
\begin{equation}
    \textsc{sel}_C \equiv \sum_{i\in[\eta]}\sum_{j<i}\sum_{k\in [d]} U^{ij}_{k}\otimes \ket{ijk}\bra{ijk}.
\end{equation}
Then, to contract over the ancilla indices and block encode the corresponding scalar operator, we prepare the uniform superposition over {the ancilla} register with,
\begin{equation}
    \textsc{prep}_C \ket{0}\ket{0}\ket{0} = \sqrt{\frac{1}{d\binom{\eta}{2}}}\sum_{i\in[\eta]}\sum_{j<i}\sum_{k\in[d]}\ket{ijk}.
\end{equation}
This therefore gives
\begin{equation}
    U_C = [I_{N^{\eta d}}\otimes\textsc{prep}^\dg_C ][\textsc{sel}_C][ I_{N^{\eta d}}\otimes\textsc{prep}_C ]
\end{equation}
as a block encoding of the diagonal terms in the convective operator. We detail the costs to construct this block encoding as \ref{thm:BE full convective}.

\begin{thm}[Block encoding many-body convective term]
    \label{thm:BE full convective}
    There exists a quantum circuit which implements a $(O(\eta^2 d N^{12}), \ceil{\log\binom{\eta}{2}} + \ceil{\log(n+1)}+ 2\ceil{\log(d)} + \ceil{\log(p)} + 6)$-block encoding of the finite difference approximation to the convective operator of  Eq. \eqref{eq:eta particle convective} with Lennard-Jones potential $V$ using $\widetilde{O}\left((\eta d p)^2 N\log^2(N)\right)$ Toffoli or simpler gates.
\end{thm}
\begin{proof}
    The proof can be obtained by direct application of \ref{lem:BE d-dim convec}. Simply apply the quantum circuit in Fig. \ref{fig:single convec term} for all $\binom{\eta}{2}$ pairs of particles. This can be accomplished by appending an additional register of $\ceil{\log\left(\binom{\eta}{2}\right)} = O(2\log(\eta))$ ancilla registers, and apply the $\textsc{sel}_k^i(D)$ and $\textsc{sel}_k^{ij}(V)$ controlled on the ancilla register marking the particles. Therefore, the total gate cost can be estimated as $O(\eta^2)$ times the cost to block encode $C^{ij}$, for a total of $\widetilde{O}(\eta^2 d^2 p^2 N\log^2(N))$ Toffoli or simpler gates. Employing a total of $\ceil{\log\binom{\eta}{2})} + \ceil{\log(n+1)}+ 2\ceil{\log(d)} + \ceil{\log(p)} + 6$ ancilla qubits and implementing the block encoding with a subnormalization factor $O(\eta^2 d N^{12})$. Now, if we express $N$ in terms of $\epsilon$ as $N \sim \epsilon^{-1/p}$, we obtain
    \begin{equation}
        \widetilde{O}\left(\eta^2 d^2 \epsilon^{-1/p}\log^2(1/\epsilon)\right)
    \end{equation}
    Toffoli or simpler gates needed to implement the desired block encoding.
\end{proof}

\section{\label{sec:Discussion} Discussion and Outlook}

The simulation of physical systems is 
{one of} the most promising use cases of future quantum computers. In this work, we have addressed the block encoding of a class of elliptic operators {that arise in such simulations, when} discretized on a Cartesian grid of $N$ grid points per dimension in $d$ spatial dimensions.
We provide explicit block encodings for the standard Laplacian operator in $d$ dimensions with Dirichlet, Neumann, and Robin boundary conditions {that require only} $O(d\log(N))$ Toffoli gates.
{However, the error convergence is restricted to low order}, i.e., $\epsilon \sim N^{-2}$(Dirichlet/Periodic) or $N^{-1}$(Neumann/Robin), using standard approaches. We provide these block encodings as linear combinations of unitaries, where the unitaries are simply the modular inc(dec)rement operators and reflections and their controlled versions, and can be implemented with $O(\log(N))$ Toffolis. {For the Dirichlet and Neumann boundary value problems,} we show how to obtain high-order convergence rates through the use of the periodic extension, so that for any $p>3$ we {can bound the global truncation error by} $\epsilon \sim O(N^{-p+1})$ for {these} boundary value problems. We also address the block encoding of operators with boundary conditions on an irregular geometry (e.g., not given as the simple Cartesian product of closed intervals), assuming {that} we can efficiently compute a membership oracle marking {to} which part of the computational domain {a specific} grid point belongs providing the same asymptotic complexity for generalized rectangular domains in $\mathbb{R}^d$.

We then showed how to block encode a many-body differential operator by first block encoding a multi-particle potential $V$ given as the sum of pair-wise interactions. We focus on the evaluation of the Lennard-Jones potential, which involves calculating terms of the form $V(r) = \frac{1}{r^{12}} - \frac{1}{r^6}$, where $r$ is the interparticle distance. This choice of potential is ubiquitous in applications in molecular dynamics, and more generally interacting kinetic theory, as described by the Fokker-Planck and Backwards Kolmogorov equations. We provide efficient quantum circuits to block encode a radially symmetric ``discrete position operator" in \ref{sub:BE pos ops}, which computes the computes the quadratic potential between two particles in $d$-dimensions. We address the singular behavior of the potential near {$r=0$} through the use of a cutoff parameter $\delta$ and polynomial approximation over the whole domain. We show an alternative technique using a piecewise polynomial approximation in conjunction with an inequality test, which may improve the convergence rate of polynomial approximations in practice. Finally, we show how to block encode the convective operator given as $\nabla V \cdot \nabla$, which describes particles interacting with a force field $\nabla V$ where $\nabla$ is the $\eta d$-dimensional gradient operator. Overall, we find that the number of Toffoli gates needed to implement these block encodings scale as $\widetilde{O}((\eta d)^2 \delta^{-1})$.  Compared to the $O(N^{\eta d})$ operations needed  to realize these matrices on a classical computer, this construction provides an exponential reduction in the number of quantum operations needed to realize the same operation on a quantum device.

With the inclusion of an interaction potential, the underlying PDE is no longer separable over the individual particles. As a result, it is unlikely that classical methods can generically and/or significantly reduce the dimensionality of the system. Thus, in the many particle setting, the potential for exponential quantum advantage is much more significant. For example, in the ubiquitous problem of solving the encoded linear system, the minimum number of classical operations generically scales with the dimensionality of linear system as $\Omega(N^{\eta d})$. On the other hand, using the constructions we provide in this work, the minimum number of non-clifford quantum operations needed to solve the same system must can be exponentially smaller, i.e. $\Omega((\eta d)^2 N\log(N))$. Of course, this discussion neglects considerations of the conditioning of the linear system, which is unlikely to scale exponentially with $\eta$. Nevertheless, this factor is present in both quantum and classical linear solvers. 

The block encodings we construct in this work are provided as linear combinations of unitaries (LCUs), additional work is {still} needed to compile these circuits into 1- and 2-qubit gates. In particular, we do not specify the construction of circuits that prepare the LCU coefficients, nor the construction of the controlled applications of the unitaries. More work is needed to address the constant factors in the complexity resulting from the controlled applications of the 1 sparse unitaries we use in the construction.  
{We note that generation of} software to automatically compile and optimize the quantum circuit implementations of these block encodings would 
be a useful development. 
{More generally, it} may also be worthwhile to explore Cartesian embedding methods \cite{cocoHighOrderFinitedifference2024,ruuthSimpleEmbeddingMethod2008a} in connection with the methods presented in Sec. \ref{subsec:irregular domains}, to construct efficient circuits for representing differential operators on curved domains. However, {we point out that} if one wishes to obtain generic high-order accuracy on irregular geometries, the $hp$-FEM approach is a much more robust framework than embedding methods. 
{Block} encoding the stiffness matrix from FEM is much more involved, and is the subject of future work.

This work represents one step toward the goal of developing efficient quantum subroutines for realistic simulation tasks {involving key differential operators} on future quantum computers. The {resulting} circuits are conceptually simple and rely on the use of well-known primitives such as shift and reflection gates to block encode the Laplacian matrix, as well as {on} more modern methods based on inequality testing and quantum signal processing to efficiently block encode a more general elliptic operator {such as} $\nabla V(\mathbf{x})\cdot \nabla$, {with} $V$ {a} scalar multi-particle potential. This work {thus} provides a blueprint for the block encoding of many-body differential equations, which is likely to be a core subroutine in future fault-tolerant applications involving the numerical solution of high-dimensional partial differential equations.  However, more work is needed to address the question of the efficiency and applicability quantum algorithms for solving PDEs beyond the Schr{\"o}dinger equation in an end-to-end fashion. {With this goal in mind, we expect these constructions will} be useful for analyzing the quantum resources needed to simulate the high-dimensional, {many-body}, systems {of interest that arise in the physical and applied sciences}.

\begin{acknowledgments}
{T.D.K., A.M.A., K.K.M., and K.B.W. were supported by the U.S. Department of Energy, Office of Science, Office of Advanced Scientific Computing Research under Award Numbers DE-SC0023273 and DE-SC0025526. T.D.K was also supported by a Siemens FutureMakers Fellowship. J.-P.L. was supported by the National Science Foundation (PHY-1818914, CCF-1729369), the NSF Quantum Leap Challenge Institute (QLCI) program (OMA-2016245, OMA-2120757), and a Simons Foundation award (No. 825053). T.D.K would also like to thank Ian Convy, Di Fang, Cory Hargus, and Torin Stetina for insightful discussions in the early portions of this work. We also thank Philipp Schleich and Torin Stetina for helpful comments on an early version of the manuscript.} 
\end{acknowledgments}

\bibliographystyle{unsrtnat}
\bibliography{Main.bib}

\appendix
\section{\label{app:circs} Circuit Implementations}
\subsection{Modular increment and decrement}
The modular left and right shift operators implement the following transformation on $N = 2^n$ states
\begin{equation}
    \begin{split}
        S^{-1}\ket{l} &= \ket{(l-1)\mod N}\\
        S^1\ket{l} &= \ket{(l+1)\mod N}.
    \end{split}
\end{equation}
In this section, we will construct the quantum circuit for the incrementer operation to provide some understanding of the underlying circuit structure.

Let's first consider the simplest example, which is just a single qubit. In this case the shift-by-one operator is just the bitflip operator, or the Pauli-$X$ gate
\begin{equation}
    \ket{b_0} \mapsto \ket{b_0\oplus1}.
\end{equation}
which is realized as the circuit

\begin{center}
    \begin{quantikz}
        \lstick{$\ket{b_0}$} & \gate{X} &\qw 
    \end{quantikz}
\end{center}

The two qubit case $\ket{b_1,b_0}\mapsto \ket{b_1\oplus b_0, b_0\oplus 1}$, which flips the state of qubit $b_1$ if the state of $b_0$ is 1 and flips $b_0$. Therefore $\ket{00} \mapsto \ket{01} \mapsto \ket{10} \mapsto \ket{11} \mapsto \ket{00}$. This circuit is realized via a controlled-X gate and an X gate:
\begin{center}
    \begin{quantikz}
        \lstick{$\ket{b_1}$} & \targ{0}\vqw{1}  &      \qw & \qw  \\
        \lstick{$\ket{b_0}$} & \ctrl{0} & \gate{X} & \qw
    \end{quantikz}
\end{center}

The three qubit case is similar. This is given by the mapping 
\begin{equation*}
\ket{b_2,b_1,b_0}\mapsto \ket{b_2\oplus b_1 b_0 , b_1\oplus b_0, b_0\oplus 1}
\end{equation*}
, which flips the 2nd qubit if the $b_0 = b_1 = 1$, flips the first qubit if $b_0 = 1$, and flips $b_0$. We can see that
$\ket{000}\mapsto \ket{001}\mapsto \ket{010} \mapsto \ket{011}\mapsto \ket{100} \mapsto \ket{101} \mapsto \ket{110} \mapsto\ket{111}$
which is realized via the circuit:
\begin{center}
\begin{quantikz}
\lstick{$\ket{b_2}$} & \targ{1} \vqw{1} &\qw            & \qw      &\qw \\  
\lstick{$\ket{b_1}$} & \ctrl{1}         & \targ{0}\vqw{1} &\qw       &\qw  \\
\lstick{$\ket{b_0}$} & \ctrl{0}         & \ctrl{0}       & \gate{X} &\qw
\end{quantikz}
\end{center}

The four qubit case is similar, the algebraic representation is 
\begin{equation*}
    \ket{b_3,b_2,b_1,b_0}\mapsto\ket{b_3\oplus b_2b_1b_0, b_2\oplus b_1b_0, b_1\oplus b_0, b_0\oplus 1}
\end{equation*}
 and the circuit is
\begin{center}
    \begin{quantikz}
        \lstick{$\ket{b_3}$} &\targ{1}\vqw{1}  &\qw &\qw            & \qw      &\qw \\  
        \lstick{$\ket{b_2}$} & \ctrl{2}& \targ{1} \vqw{1} &\qw            & \qw      &\qw \\  
        \lstick{$\ket{b_1}$} & \ctrl{1}& \ctrl{1}         & \targ{0}\vqw{1} &\qw       &\qw  \\
        \lstick{$\ket{b_0}$} & \ctrl{0}& \ctrl{0}         & \ctrl{0}       & \gate{X} &\qw
    \end{quantikz}  
\end{center}
Carrying on in a similar fashion, for larger numbers of qubits will prepare the $R$ operator for $n-$qubits. The implementation for the left-shift-by-one (i.e. decrement) operator is similar, and can be obtained by replacing the control on 1 with a control on 0 or equivalently conjugating the increment circuit by with Pauli $X$ gates on all the wires. The total circuit depth is $\mathcal{O}(n)$. However, the multi-control Toffoli gates used in this scheme need to be compiled down to Toffoli gates, which can then be further compiled to one and two qubit gates. In \cite{barencoElementaryGatesQuantum1995}, it is shown that the direct compilation of this circuit will require $O(n^2)$ Toffoli or simpler gates. However, we can use the scheme provided in \cite{ConstructingLargeIncrement} to implement this quantum circuit using $O(n)$ Toffoli gates, at the cost of an ancilla qubit.

Notice that when we desire to shift by any non-negative power of 2, say $0 \leq j \leq n$, the above circuit simplifies to only acting on the $n-j$ high bits of the register. Therefore, to implement the shift by constant circuit, it is sufficient to represent that constant as a signed combination (i.e. coefficients $\pm 1$) of positive powers of $2$ less than or equal to $n$. Therefore, combining the above increment-by-one circuit (which is just incrementing by $2^0$), with a series of increments or decrements by powers of $2$ to obtain the desired operation. Since in practice $j$ is typically a small constant that does not scale with $n$, this procedure is quite efficient. 

\subsection{Reflection Operators}
\begin{lem}[Cost to synthesize diagonal reflection unitary]
\label{lem:ref mats circuit}
    Let $A \subset [N]$ mark states in the standard computational basis. Let $R_A$ be reflection about the states in $A$, that is
    \begin{equation}
        R_A = -2\sum_{j\in A}\ket{j}\bra{j} + I
    \end{equation}
    which is just
    \begin{equation}
        R_A\ket{i} = \begin{cases}
            1 & i \notin A\\
            -1 & i \in A
        \end{cases}
    \end{equation}
    Then the number of elementary gates to implement $R_A$ is $O(|A| \mathcal{C}(C^n(Z)))$ where $\mathcal{C}(C^n(Z))$ is the cost to implement the $n$-qubit controlled $z$-gate.
\end{lem}
\begin{proof}
    Let $a \in A$ and $a\in\{0,1\}^n$ is an $n-$bit string so that we may write $a = \sum_{j=0}^{n-1} a_j 2^j \simeq \{a_0, a_1, \ldots, a_{n-1}\} \equiv b$. Let $S[a] = \{j; a_j = 0 \forall a_j \in a\}$ be the set of all unset bits of $a$. Then we let $X_{S[a]} = \bigotimes_{i\in[n]}X^{\neg a_i}$ be the operator that performs an $X-$gate on every qubit that corresponds to an unset bit in the bitstring $a$. Additionally, let $C^{n}(Z)$ be the n-fold controlled $Z$ operation,
    \begin{equation}
        C^{n}(Z) = \sum_{i \in [N-1]}\ket{i}\bra{i} - \ket{N-1}\bra{N-1}.
    \end{equation}
    Then $\bra{d_0}\bra{d_1}\cdots\bra{d_{n-1}}X_{S[a]}C^{n}(Z)X_{S[a]}\ket{d_0}\ket{d_1}\cdots\ket{d_{n-1}}$ satisfies
    \begin{align*}
        &\bra{d_0}\bra{d_1}\cdots\bra{d_{n-1}}X_{S[a]}C^{n}(Z)X_{S[a]}\ket{d_0}\ket{d_1}\cdots\ket{d_{n-1}}\\
        &=\bra{d_0\oplus \neg a_0}\bra{d_1\oplus \neg a_1}\cdots\bra{d_{n-1}\oplus \neg a_{n-1}}C^{n}(Z)\ket{d_0\oplus \neg a_0} \ket{d_1\oplus \neg a_1}\cdots \ket{d_{n-1}\oplus \neg a_{n-1}}\\
    \end{align*}
    But then, since $\bra{i}C^n(z)\ket{i} = -1 \iff i = N-1$, implies that in order for the input state $\ket{d} = \ket{d_0}\cdots\ket{d_{n-1}}$ to be in the marked subspace, it must satisfy:
    \begin{equation}
    \begin{aligned}
        d_0\oplus \neg a_0 &= 1\\
        d_1 \oplus \neg a_1 &= 1\\
        &\vdots\\
        d_{n-1}\oplus \neg a_{n-1} &= 1
    \end{aligned}
    \end{equation}
    which implies
        \begin{equation}
    \begin{aligned}
        d_0 &= a_0\\
        d_1 &= a_1\\
        &\vdots\\
        d_{n-1} &= a_{n-1}
    \end{aligned}
    \end{equation}
    therefore, the operations above prepared the reflection about the marked state $a$. 

    We can form a product of these reflections to get a reflection over all the elements in a subspace marked by elements of a set. Since the number of elements in the set is $|A|$, we expect this process needs to be performed for $O(|A|)$ times. Resultantly, the cost to implement this operation is $O(|A| \mathcal{C}(C^n(z)))$ where $\mathcal{C}(C^n(z)) = O(n)$ is the cost to implement the controlled-Z on n-qubits in terms of Toffoli and T gates.
\end{proof}

\section{\label{app:BE} Block Encoding Costs}
\begin{thm}[Block encoding $d-$dimensional Laplacian with Dirichlet Boundary]
\label{thm:BE LCU dir apdx}
    Let $L$ be the $d-$dimensional finite difference Laplacian with a $p = 2a+1$ point scheme on $N$ grid points. And let $A = A_0 \cup A_1 \cdots \cup \cdots A_{d-1}$ be the set of all boundary points over the $d$-dimensional domain with each $A_i \subset [N]$. We additionally assume that each $A_i$ is a connected set in the sense that $\forall x,y \in A_i$, $\exists i \leq |A_i|$ such that $x\pm i \equiv y \mod N $ and we also assume that each $A_i$ can be specified without dependence on the coordinates in $A_j$ for every $j\neq i$. 
    
    We define the shifted set
    \begin{equation}
        A_{i+c} = \{j + c \mod N: j\in A_i\}.
    \end{equation}
    Then the $1d$ Laplacian $L_d$ with boundary nodes $A_d$ can be specified via the LCU
    \begin{equation}
        \frac{1}{2}\sum_{j=-a}^a r_j\left(S^j + S^jR_{A_d-j}\right) + \frac{1}{2}(I + R_{A_d}).
    \label{eq:Laplacian LCU appdx}
    \end{equation}
    
    And the $d$-dimensional Laplacian can be encoded using the direct sum
    \begin{equation}
        L = \bigoplus_{i=0}^{d-1} L_i,
    \end{equation}
    where $L_i$ is the Laplacian formed by separable connected Dirichlet boundary conditions in each direction $i$.
\end{thm}
\begin{proof}
    Our goal is to show that Eq. \eqref{eq:Laplacian LCU appdx} results in a matrix acting as the identity on the boundary points and the Laplacian on the interior points. Then, we enforce the assumed condition that the boundary nodes can be specified in each dimensions independently to combine these one-dimensional operators into the appropriate $d-$dimensional operator.

    The finite-difference, periodic $L$ operator on $N$ nodes with is expressible as the sum
    \begin{equation}
        L = \sum_{i=-a}^{a} r_i S^i
    \end{equation}
    where the coefficients $r_i$ are given by Eq. \eqref{eq:FD coeffs lap} and $S^i$ are the shift operators. Now, consider $A_d$ a set of boundary points on the $d$th spatial coordinate. Then,
    \begin{align*}
        &\sum_{j=-a}^a(r_jS^j + r_jS^jR_{A_d-j})\\
        &=\sum_{j=-a}^a(r_jS^j + r_jS^j(-2\sum_{k\in A_d} \ket{k-j}\bra{k-j} -I))\\
        &=\sum_{j=-a}^ar_j \sum_{i\in[N]}(\ket{i+j}\bra{i} + \ket{i+j}\bra{i}(-2\sum_{k\in A_d} \ket{k-j}\bra{k-j} + I))\\
    \end{align*}
    \begin{align*}
        &\sum_{i\in[N]}(-2\sum_{k\in A_d} \ket{i+j}\bra{i}\ket{k-j}\bra{k-j} + \ket{i+j}\bra{i}) \\
        &=-2\sum_{i\in[N]}\sum_{k\in A_d} \ket{i+j}\delta_{i,k-j}\bra{k-j} + \sum_{i\in[N]}\ket{i+j}\bra{i} \\
        &=-2\sum_{k\in A_d} \ket{k}\bra{k-j} + \sum_{i\in[N]}\ket{i+j}\bra{i}\\
    \end{align*}
    \begin{align*}
        &\sum_{i\in[N]}\ket{i+j}\bra{i} +  \ket{i+j}\bra{i}(-2\sum_{k\in A_d} \ket{k}\bra{k-j} + I))\\
        &= \sum_{i\in[N]}\ket{i+j}\bra{i} - 2\sum_{k\in A_d} \ket{k}\bra{k-j} + \sum_{i\in[N]}\ket{i+j}\bra{i}\\ 
        &= 2\sum_{i\in[N]}\ket{i+j}\bra{i} - 2\sum_{k\in A_d} \ket{k}\bra{k-j}\\
        &\propto \sum_{i\in[N]}\ket{i+j}\bra{i} - \sum_{k\in A_d} \ket{k}\bra{k-j}\\
        &=
        \begin{cases}
          0 & \forall \text{ } k \in A_d :\hspace{.2cm} \ket{i+j}\bra{i} = \ket{k}\bra{k-j}\\
          1 & \text{else}
        \end{cases}
    \end{align*}
    Then,
    \begin{align*}
        &\bra{m}\sum_{i\in[N]}\ket{i+j}\bra{i} - \sum_{k\in A_d} \ket{k}\bra{k-j}\ket{n}\\
        &=\sum_{i\in[N]}\delta_{i+j,m}\delta_{i,n} - \sum_{k\in A_d} \delta_{m,k}\delta_{n,k-j}\\
        &=\delta_{n+j,m} - \delta_{m,A_d} \delta_{n,m-j}\\
        &=\delta_{n+j,m} - \delta_{m,A_d} \delta_{n+j,m}\\
        &=\begin{cases}
            0 & m \in A_d \text{ and } n+j = m\\
            1 & m \notin A_d \text{ and } n+j=m
        \end{cases}
    \end{align*}
    which is the desired behavior. 
    
    So, we may now write:
    \begin{align*}
        &\sum_{j=-a}^ar_j\left(\sum_{n,m\in[N]}\delta_{n+j,m} - \sum_{n\in[N]}\sum_{m\in A_d}\ \delta_{n+j,m}\right)\\
        &=\sum_{j=-a}^ar_j\left(S^j - \sum_{n\in[N]}\sum_{m\in A_d}\ \delta_{n+j,m}\right)\\
        &=\sum_{j=-a}^ar_j\left(S^j - \sum_{n\in[N]}\sum_{m\in A_d} \delta_{n+j,m}\right)
    \end{align*}
    which more plainly now, we can see that this serves to zero out the entries on the rows supported in $A_d$. Now we turn to the particular case of the diagonal:
    \begin{align*}
        r_0[I - \sum_{m\in A_d}\delta_{m,m}]_{jj} = \begin{cases}
            r_0 & m \notin A_d\\
            0 & m \in A_d.
        \end{cases}
    \end{align*}
    Now, if we add another projector onto $A_d$ we find,
    \begin{equation}
        [r_0(I - \sum_{m\in A_d}\delta_{m,m}) + \Pi_{A_d}]_{ii} = \begin{cases}
            r_0 & i \notin A_d\\
            1 & i \in A_d
        \end{cases}.
    \end{equation}

    Then, we can decompose the projector $\Pi_{A_d} = I - \Pi_{A_d^c}$ as 
    \begin{equation}
        \Pi_{A_d} = \frac{1}{2}\left(I - R_{A_d^c}\right) = \frac{1}{2}\left(I + R_{A_d}\right)
    \end{equation}
    to obtain the final LCU
    \begin{equation}
        L_d =\frac{1}{2}\sum_{j=-a}^a r_j\left(S^j + S^jR_{A_d-j}\right) + \frac{1}{2}(I + R_{A_d}).
        \label{eq:LCU for BCs}
    \end{equation}
    Resultantly, the final matrix has the form
    \begin{equation}
    L_d = 
        \begin{pmatrix}
            {L}^{(l)}&  & \\
             & I_{|A_d|} & \\
             & & {L}^{(r)}
        \end{pmatrix}
    \end{equation}
    as desired.
\end{proof}

\begin{thm}[Cost to block encode the LCU given in Eq. \eqref{eq:Lap LCU} ]
\label{thm:BE Diffusion Cost appdx}
    The number of elementary gates needed to implement the LCU given by Eq. \eqref{eq:Lap LCU} is $O(p|A|\mathcal{C}(C^n(Z)))$ Toffoli. Furthermore, the above implements an $(O(1), \log(2p),0)$ block encoding of the desired matrix.
\end{thm}
\begin{proof}
    The linear combination provided via Eq. \eqref{eq:LCU for BCs} requires $2p$ shift matrices and $p$ reflection matrices. Using the construction from \cite{ConstructingLargeIncrement} gates, the shift operators can be encoded using $O(n)$ Toffoli-or-smaller gates with an additional ancilla. We assume these to be the bottleneck for implementing these operations, so we will count the complexity in terms of the total number of Toffoli gates needed. Therefore, the cost to implement these shift gates is $O(np)$ Toffoli or simpler gates.

    The reflection matrices are diagonal matrices with $\pm 1$ on the diagonal, with $-1$'s on rows who's indices correspond to states of the subspace being reflected over,
    \begin{equation}
        R_{A}\ket{j} =\begin{cases}
            \ket{j} & j \notin A\\
            -\ket{j} & j \in A
        \end{cases}
    \end{equation}
    as given in Lemma \ref{lem:ref mats circuit}, the cost to implement the quantum circuit for the reflection is $O(|A| \mathcal{C}(C^n(Z))p)$. So that the total number of Toffoli or simpler operations needed is $O(p|A|\mathcal{C}(C^n(Z))) = O(np)$. 

    The subnormalization factor is the $1-$norm of the coefficients in the LCU.
    \begin{equation}
        \lambda = \sum_{j=-a}^a |r_j| + 1 \leq \frac{2\pi^2}{3} + 1 = O(1)
    \end{equation}
    where the first inequality can be found in Lemma 6 of Ref \cite{kivlichanBoundingCostsQuantum2017}. Then, in $d$-dimensions, using $O(\log(p) + \log(d))$ ancilla qubits to form the linear combination
    \begin{equation}
        \frac{1}{d}\sum_{i\in[d]}U_{L_j}
    \end{equation}
    where $U_{L_j}$ is the block encoding of the direction $j$ Laplacian. 

    So in total, this requires $O(dp+\log(d))$ ancilla, $\alpha = O(\frac{1}{d})$, a number of Toffoli or smaller gates that scales as $O(dnp\mathcal{C}(C^n(Z)) |A|)$ and in time $O(np\mathcal{C}(C^n(Z)))$. And similarly for $\eta$ particles another multiplicative factor of $\eta$; $\alpha = O(\frac{1}{\eta d}) < \frac{\pi^2}{\eta d}$, $T = O(\eta dnp\mathcal{C}(C^n(Z)) |A|) $, but add no cost to overall simulation time as they can be performed in parallel. 
\end{proof}

\begin{thm}[Proof of circuit in Fig. \ref{fig:BE-gradV}]
\label{thm:gradV output}
    There exists a set of phases $\phi_i \in \mathbb{R}$ so that the quantum circuit in Fig. \ref{fig:BE-gradV} encodes a superposition of evaluations of the approximate Lennard-Jones potential.
\end{thm}
\begin{proof}
    By lemma \ref{lem:qsp}, for any polynomial $P$ and $Q$ satisfying the conditions of the lemma, we are promised the existence of phases $\phi_i$ that encode the polynomial the QSP ansatz. The existence and rate of convergence of the approximating polynomial satisfying these assumptions is given by lemma \ref{lem: convRate for invPow}. Now, we verify that the quantum circuit in Fig. \ref{fig:BE-gradV} produces the desired output. The proof is obtained from direct computation of the operations. After the first layer of gates, the input quantum state is transformed to,
    \begin{align*}
        \ket{0}_{dim}\ket{0}_{coeff}\ket{1}_{phase}\ket{0}_{anc}\ket{b}_{sys} &\overset{D(d)\otimes \textsc{prep}_D\otimes e^{i\phi_0 Z}}{\longrightarrow} \frac{e^{i\phi_{0}}}{\sqrt{\beta d}}\sum_{k\in[d]}\sum_{1\leq |j| \leq a} \sqrt{c_j} \ket{k}_{dim}\ket{j}_{coeff}\ket{1}_{phase}\ket{0}_{anc}\ket{b}_{sys}.
    \end{align*}
    Then, applying the controlled block-encoding $\left({R^{ij}_{kl}}\right)^2$, we obtain 
    \begin{align*}
        \frac{e^{i\phi_{0}}}{\sqrt{\beta d}}\sum_{k\in[d]}\sum_{1\leq |j| \leq a} \sqrt{c_j}\left(\ket{k}_{dim}\ket{j}_{coeff}\ket{1}_{phase}\ket{0}_{anc}\left({R^{ij}_{kl}}\right)^2\ket{b}_{sys} + \ket{\perp}\right).
    \end{align*}
    Let us now consider the action of the rest of the circuit on single basis state,
    \begin{align*}    e^{i\phi_{0}}\left(\ket{k}_{dim}\ket{j}_{coeff}\ket{1}_{phase}\ket{0}_{anc}\left({R^{ij}_{kl}}\right)^2\ket{b}_{sys} + \ket{\perp}\right).         
    \end{align*}
    Due to the controlled application of the block encoding, the $dim$ and $coeff$ registers are unchanged and can be ignored for this part. Therefore, this portion of the circuit corresponds to applying QSP to the block encoding of the shifted difference of positions operator  
    \begin{align*}
    e^{i\phi_{0}}\left(\ket{1}_{phase}\ket{0}_{anc}\left({R^{ij}_{kl}}\right)^2\ket{b}_{sys} + \ket{\perp}\right)\overset{qsp}{\longrightarrow} e^{i\phi_{0}}\left(\ket{1}_{phase}\ket{0}_{anc}\widetilde{V}\left(\left({R^{ij}_{kl}}\right)^2\right)\ket{b}_{sys} + \ket{\perp}\right).
    \end{align*}
    Now, with the $coeff$ and $dim$ registers, we have
    \begin{align*}
        e^{i\phi_{0}}\left(\frac{1}{\sqrt{\beta d}}\sum_{k\in[d]}\sum_{1\leq |j| \leq a} \ket{k}_{dim}\ket{j}_{coeff}\ket{1}_{phase}\ket{0}_{anc}\sqrt{c_j}\widetilde{V}\left(\left({R^{ij}_{kl}}\right)^2\right)\ket{b}_{sys} + \ket{\perp}\right),
    \end{align*}
    and performing $\textsc{prep}_D^\dagger$ on the $coeff$ register and $X$ on the $phase$ register, we have 
    \begin{align*}
        e^{i\phi_{0}}\left(\frac{1}{\beta\sqrt{d}}\sum_{k\in[d]} \ket{k}_{dim}\ket{0}_{coeff}\ket{0}_{phase}\ket{0}_{anc}\sum_{1\leq |j| \leq a}c_j\widetilde{V}\left(\left({R^{ij}_{kl}}\right)^2\right)\ket{b}_{sys} + \ket{\perp}\right),
    \end{align*}
    which we can simplify as
    \begin{align*}
        &e^{i\phi_{0}}\left(\frac{1}{\beta\sqrt{d}}\sum_{k\in[d]} \ket{k}_{dim}\ket{0}_{coeff}\ket{0}_{phase}\ket{0}_{anc}\partial_k^i\widetilde{V}\left(\left({R^{ij}_{kl}}\right)^2\right)\ket{b}_{sys} + \ket{\perp}\right),\\
    \end{align*},
    so that upon measuring the $dim$ $coeff$ $phase$ and $anc$ registers in the zero state we have
    \begin{equation}
     e^{i\phi_{0}}\left(\frac{1}{\beta\sqrt{d}}\sum_{k\in[d]} \ket{k}_{dim}\partial_k^i\widetilde{V}\left(\left({R^{ij}_{kl}}\right)^2\right)\ket{b}_{sys} + \ket{\perp}\right),
    \end{equation}
    which produces the desired output, a finite difference approximation to the gradient.
\end{proof}



\end{document}